\documentclass[12pt, onecolumn]{IEEEtran}
\usepackage{amsmath}
\usepackage{amssymb}
\usepackage{graphicx}
\usepackage{verbatim}
\usepackage{cite}

\newtheorem{lemma}{Lemma}
\newtheorem{theorem}{Theorem}

\newcommand{\sinr}{{\rm SINR}}
\newcommand{\snr}{{\rm SNR}}

\newcommand{\Nr}{{n_{\mathrm{\scriptscriptstyle R}}}}

\newcommand{\bh}{{\bf h}}
\newcommand{\bv}{{\bf v}}
\newcommand{\E}{\mathbb{E}}
\newcommand{\prob}{\mathbb{P}}
\newcommand{\Pout}{\textrm{P}_{\textrm{out}}}

\begin{document}
\title{Multi-antenna Communication in Ad Hoc Networks: Achieving MIMO Gains with SIMO Transmission}

\author{Nihar Jindal, Jeffrey G. Andrews and Steven Weber
\thanks{N. Jindal is with the Univ. of Minnesota, J. Andrews is with the Univ. of Texas at Austin, S. Weber is with Drexel Univ.  The contact author is N. Jindal, nihar@umn.edu.
Preliminary results appeared at ICC 2009 \cite{JinAnd09}.}}

\maketitle

\begin{abstract}
The benefit of multi-antenna receivers is investigated in
wireless ad hoc networks, and the main finding is that network throughput
can be made to scale linearly with the number of receive antennas $\Nr$
even if each transmitting node uses only a single antenna.
This is in contrast to a large body of prior work in single-user, multiuser, and ad hoc wireless networks
that have shown linear scaling is achievable when multiple receive and transmit antennas (i.e., MIMO
transmission) are employed,
but that throughput increases logarithmically or sublinearly with $\Nr$ when only a single transmit antenna (i.e., SIMO transmission)
is used.
The linear gain is achieved by using the receive degrees of freedom to simultaneously suppress interference and
increase the power of the desired signal, and exploiting the subsequent performance benefit to increase
the density of simultaneous transmissions instead of the transmission rate.
This result is proven in the transmission capacity framework, which presumes
single-hop transmissions in the presence of randomly located interferers, but it is also
illustrated that the result holds under several relaxations of the model,
including imperfect channel knowledge, multihop transmission, and regular
networks (i.e., interferers are deterministically located on grids).
\end{abstract}

\section{Introduction}

Multiple antenna communication has become a key component of virtually
every contemporary high-rate wireless standard (LTE, 802.11n,
WiMAX).  The theoretical result that sparked the intense
academic and industrial investigation of MIMO (multiple-input/multiple-output) communication was the
finding that the achievable throughput of a
point-to-point MIMO channel scales \textit{linearly} with the
minimum of the number of transmit and receive antennas
\cite{FosGan98,Tel99}.
Linear scaling in the number of transmit antennas can also be achieved in point-to-multipoint (broadcast/downlink)
channels \cite{CaiSha03} even if each receiver has only a single antenna, or in the number
of receive antennas in multipoint-to-point (multiple access/uplink) channels even if each
transmitter has a single antenna \cite{Tel99}.  In these two cases, the linear gains are enabled by the
simultaneous transmission or reception of multiple data streams.

\subsection{Overview of Main Results}

In this paper we are interested in the throughput gains that
multiple antennas can provide in \textit{ad hoc networks}, rather
than in channels with a common transmitter and/or receiver.  If multiple antennas are added at each node in the network
and point-to-point MIMO techniques are used to increase the rate of
every individual link (i.e., every hop in a multi-hop route) in the network,
then network-wide throughput naturally increases
linearly with the number of antennas per node.  Similarly, based on
the quoted point-to-multipoint and multipoint-to-point MIMO results,
linear scaling is also expected to occur if
nodes are capable of sending or receiving multiple streams.  The main finding of this paper is that network-wide
throughput can be increased linearly with the number of receive antennas, even
if only a single transmit antenna is used by each node (i.e., single-input, multiple-output, or SIMO, communication), and each node
sends and receives only a single data stream. Furthermore, this gain
is achievable using only linear receive processing and does not require any transmit
channel state information (CSIT).

The main result is obtained by considering an ad hoc network in
which transmitters are randomly located on the plane
according to a 2-D homogeneous Poisson point process with a particular spatial density.  We consider
a desired transmit-receive pair separated by a fixed distance, and
experiencing interference (assumed to be treated as noise) from all other active transmitters.
The received signal and interference are functions of path-loss
attenuation and fading, and we assume that for a transmission to
be successful, it must be detected with an SINR larger than a defined
threshold $\beta$.  The primary performance metric is the maximum spatial density of transmitters/interferers that can be supported such that the outage probability $\prob[\sinr < \beta]$ is no larger than an
outage constraint $\epsilon$, and our particular interest
is in quantifying how quickly this maximum density increases with the number of receive antennas $\Nr$
(with $\beta$ and $\epsilon$ fixed).

If receive beamforming is performed, the receive degrees of freedom can be used to either increase the
power of the desired signal (i.e., for array gain)
or suppress interference, and these competing objectives are optimally balanced by the SINR-maximizing MMSE (minimum mean-square error)
filter.  Using the MMSE outage probability lower bound in \cite{Gao_Smith}, we develop an upper bound on the maximum
density allowable with an MMSE receiver.  In conjunction, we develop a density lower bound by analyzing the performance of
a novel suboptimal \emph{partial zero forcing} (PZF) receiver, which uses an explicit fraction of the  degrees of freedom
for array gain and the remainder for interference cancellation.
By showing that both the lower and upper bound are linear in $\Nr$, we can conclude that the optimum transmit
density is $\Theta(\Nr)$, and we demonstrate that this allows well known metrics like the transmission capacity \cite{WebYan05,WebAndJinTut},
transport capacity \cite{GupKum00,XueKum06}, and the expected forward progress \cite{Bac06} to all increase as $\Theta(\Nr)$ as well.

\subsection{Related Work}

In addition to the large body of work on point-to-point and multiuser MIMO systems,
this paper is also related to several prior works that have
studied the use of multiple receive antennas in ad hoc networks with
Poisson distributed transmitters.   References
\cite{HunAnd08} and \cite{HuaAndSub} considered precisely the same
model as this paper, but studied the performance of slightly different
receiver designs that turn out to yield very different performance.
In \cite{HunAnd08} the receive filter is chosen according to the
maximal ratio criterion and thus only provides array gain\footnote{Although reference
\cite{HunAnd08} also considers strategies involving multiple
transmit antennas, here we have mentioned only the directly relevant
single transmit/multiple-receive scenario results.}, while
\cite{HuaAndSub} considers the other extreme where the $\Nr$
antennas are used to cancel the strongest $\Nr - 1$ interferers but
no array gain is obtained.  Both receiver designs achieve only a
\textit{sublinear} density increase with $\Nr$, at growth rates
$\Nr^{2/\alpha}$ and $\Nr^{1-2/\alpha}$,
respectively.  On the other hand, we show that linear density scaling is achieved
if a fraction of the receive degrees of freedom are used for array gain and the other fraction for interference
cancellation, or an MMSE receiver is used.
In \cite{GovBli07} the performance of the MMSE receiver is investigated for a network
in which the density is fixed and additional receive antennas are used to increase the per-link SINR/rate.
It is shown that the average per-link SINR increases
with $\Nr$ as $\Nr^{\alpha/2}$, where $\alpha$ is the path-loss exponent, which translates
into only a logarithmic increase in per-link rate and overall system throughput.
In contrast, we study the rather different setting in which the per-link rate is fixed and the
density increases with $\Nr$, and we show that exploiting the antennas to increase density
instead of rate provides a significantly larger (i.e., linear versus logarithmic in $\Nr$)
end-to-end benefit (c.f. Section \ref{sec:efp}).

Early work on characterizing the throughput gains from MIMO in ad hoc networks
includes \cite{RamRed05,Ram01,CheGan05,YeBlu04} although these generally primarily employed
simulations, while more recently \cite{StaPro07a,HunAnd_Spaswin08,LouCol07} used tools similar to those
used in the paper and developed by the present authors.  However, none of these works have characterized the
maximum throughput gains achievable with receiver processing only.

\vspace{.2in}

The remainder of the paper is organized as follows.  The system model and
key metrics are described in Section \ref{sec:model}.  The
main results are derived in Section \ref{sec:main}, and
then various extensions and relaxations of the model are considered in
Section \ref{sec:extensions}: in all of these diverse permutations we observe that
the linear scaling result still holds.  We conclude in Section \ref{sec:conc}.

\section{System Model and Metrics}
\label{sec:model}

We consider a network in which the set of active transmitters are located
according to a 2-D homogeneous Poisson point process (PPP) of
density $\lambda$ (transmitters/${\textrm m}^2$).  Each transmitter
communicates with a receiver a distance $d$ meters away from it, where
it is assumed that each receiver is randomly located on a circle of radius $d$ centered around its associated
transmitter. Note that the receivers are not a part of the
transmitter PPP.  Each transmitter uses only one antenna, while each
receiver has $\Nr$ antennas. This setup is depicted in Fig. \ref{fig:model}.
The Poisson model is reasonable for uncoordinated networks, such as those
using ALOHA.\footnote{In Section \ref{sec:geometry} we briefly examine
a regular interferer geometry that is a more appropriate model for networks employing more sophisticated
random access techniques.}

\begin{figure}
\centering
\includegraphics[width=2.6in]{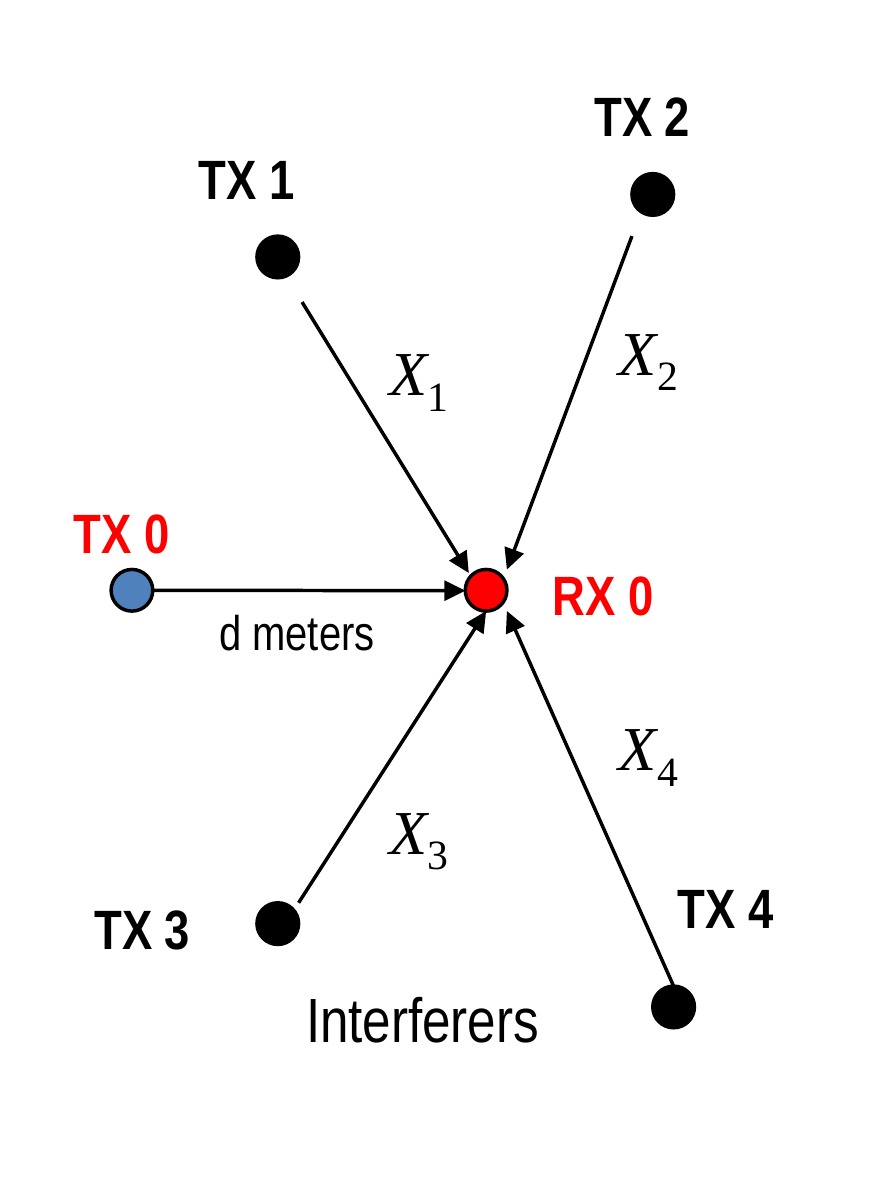}
\caption{Example transmit-receive pair with four nearby interferers shown.  In addition
to the distances $d$ (desired) and $X_i$ (interferer $i$) shown, there is a fading
value $h_i$.} \label{fig:model}
\end{figure}

By the stationarity of the Poisson process we can consider the
performance of an arbitrary TX-RX pair, which we refer to as TX0 and
RX0.  From the perspective of RX0, the set of interferers (which is
the entire transmit process with the exception of TX0) also form a
homogeneous PPP due to Slivnyak's Theorem; see \cite{WebAndJinTut}
for additional discussion of this point and further explanation of
the basic model.
As a result, network-wide performance is characterized by the
performance of a single TX-RX pair, separated by $d$ meters and
surrounded by an infinite number of interferers located on the infinite 2-D plane
according to a homogeneous PPP with density $\lambda$
interferers/$\textrm{m}^2$.

Assuming a path-loss exponent of
$\alpha$ ($\alpha > 2$) and a frequency-flat channel, the
$\Nr$-dimensional received signal ${\bf y}_0$ at RX0 is given by:
\begin{eqnarray}
{\bf y}_0 = d^{-\alpha/2} {\bh_0} u_0 + \sum_{i \in \Pi(\lambda)}
|X_i|^{-\alpha/2} \bh_i u_i + {\bf z}
\end{eqnarray}
where $|X_i|$ is the distance to the $i$-th transmitter/interferer (as determined by
the realization of the PPP),
$\bh_i \in \mathcal{C}^{\Nr \times 1}$ is the vector channel from the $i$-th transmitter to RX0,
${\bf z} \in \mathcal{C}^{\Nr \times 1}$ is complex Gaussian noise with covariance $\eta {\bf I}$,
and $u_i$ is the power-$\rho$ symbol transmitted by the $i$-th
transmitter. Consistent with a rich scattering environment, we
assume that each of the vector channels $\bh_i$ have
iid unit-variance complex Gaussian components, independent across transmitters.

\subsection{Performance Metrics}

If a unit norm receive filter $\bv_0$ is used, the resulting
signal-to-interference-and-noise ratio is:
\begin{eqnarray}
\sinr = \frac{  \rho d^{-\alpha} |\bv_0^{\dagger} \bh_0 |^2} { \eta
+ \sum_{i \in \Pi(\lambda)} \rho |X_i|^{-\alpha} | \bv_0^{\dagger}
\bh_i |^2}.
\end{eqnarray}
Without loss of generality, we assume that the distances $|X_i|$ are
in increasing order in order to take advantage of the
property that the ordered squared-distances $|X_1|^2, |X_2|^2,
\ldots$ follow a 1-D Poisson point process with intensity $\pi
\lambda$ \cite{Hae05}.  To simplify notation we define the
constant $\snr \triangleq \frac{  \rho d^{-\alpha}}{\eta}$ as the
interference-free signal-to-noise ratio, which allows us to write:
\begin{eqnarray}
\sinr = \frac{ |\bv_0^{\dagger} \bh_0 |^2} { \frac{1}{\snr} +
d^{\alpha} \sum_{i \in \Pi(\lambda)} |X_i|^{-\alpha}
|\bv_0^{\dagger} \bh_i |^2}.
\end{eqnarray}

The received SINR depends on the interferer locations and the vector
channels, both of which are random.  The outage probability\footnote{The implicit assumption is that the channels and the set of active interferers are constant for the
duration of a packet transmission but generally vary across transmissions.   As a result, the outage
probability accurate approximates the packet error probability experienced by each node; by the
stationarity of the process, it also approximates the network-wide packet error probability.}
 with respective to
 a pre-defined SINR threshold $\beta$ is:
\begin{eqnarray}
\Pout (\lambda) = \prob \left[ \sinr \leq \beta \right],
\end{eqnarray}
which clearly is increasing  in $\lambda$.  It is often desirable from a system perspective
to maintain a constant outage level $\epsilon$ (e.g., to ensure that
higher-layer reliability mechanisms are appropriately utilized), and thus the performance metric
of interest is $\lambda_{\epsilon}$, the maximum interferer density
such that the outage does not exceed $\epsilon$:
\begin{equation}
\lambda_{\epsilon} \triangleq \max_{\lambda} \{ \lambda: \Pout (\lambda) \leq \epsilon \}.
\end{equation}

The \emph{transmission capacity} is the number of successful transmissions per unit area, and
can be defined as $\lambda_{\epsilon} (1-\epsilon) b$, where $b=\log_2(1+\beta)$ is the
data rate assuming a good channel code;
this definition is akin to area spectral efficiency (ASE).
Discussion of how this metric translates to end-to-end metrics such as \textit{transport capacity}
is provided in Section \ref{sec:efp}.\footnote{Although $\lambda_{\epsilon}$ depends on the design parameters $\epsilon$, $\beta$, and $\Nr$, in this work we are interested in the behavior of
$\lambda_{\epsilon}$ with respect to $\Nr$ while the other parameters are kept fixed.  (A justification for keeping $\epsilon$ fixed has already been put forth, and justification for not increasing $\beta$ is provided in Section \ref{sec:efp}.) Thus, we henceforth denote $\lambda_{\epsilon}$ as a function of $\Nr$.}

\subsection{Receive Filters}

The SINR and maximum density depend critically on the receive filter that is used.
The receive filter can be used to either boost the power of the
desired signal (by choosing $\bv_0$ in the direction
of $\bh_0$) or to cancel interference, or some combination of the two.
In this paper we consider the MMSE receiver, which optimally balances
signal boosting and interference cancellation and maximizes the SINR, as well as
a sub-optimal partial zero-forcing receiver, which uses a specified number of degrees of
freedom for signal boosting and the remainder for cancellation.
We assume that the receive filter is chosen based upon knowledge of the
signal channel ${\bf h}_0$ and the interfering channels $\{{\bf h}_i\}_{i=1}^{\infty}$;
this optimistic assumption of perfect receiver channel state information (CSI) is scrutinized in
Section \ref{sec:csi}.

\textbf{MMSE Receiver.} From basic results in estimation theory, the MMSE receive filter is given, in unnormalized form, by:
\begin{eqnarray} \label{eq-mmse_filter}
{\bf v}_0 = {\boldsymbol \Sigma}^{-1}  {\bf h}_0.
\end{eqnarray}
where ${\boldsymbol \Sigma}$ is the spatial covariance of the interference plus noise
\begin{equation}
{\boldsymbol \Sigma} \triangleq \frac{1}{\snr} {\bf I} + d^{\alpha} \sum_{i \in
\Pi(\lambda)} |X_i|^{-\alpha} {\bf h}_i {\bf h}_i^{\dagger}.
\end{equation}
Note that ${\boldsymbol \Sigma}$ is the covariance matrix conditioned on the interferer channels and distances $\{ \bh_i \}_{i=1}^{\infty}$ and  $\{|X_i|\}_{i=1}^{\infty}$.
Amongst the set of all possible
receive filters ${\bf v}_0$, the MMSE filter maximizes the received SINR.  Its corresponding value is:
\begin{equation} \label{sinr-mmse}
\sinr^{\textrm{mmse}} =
{\bf h}_0^{\dagger} {\boldsymbol \Sigma}^{-1}  {\bf h}_0.
\end{equation}
The corresponding outage probability and maximum density are denoted as $\Pout^{\textrm{mmse}}(\lambda)$
and $\lambda_{\epsilon}^{\textrm{mmse}}(\Nr)$.

\textbf{Partial Zero Forcing Receiver.}  We also study a suboptimal receiver that explicitly cancels interference from nearby transmitters
while using the remaining degrees of freedom to boost the power of the desired signal.
More specifically, the filter $\bv_0$ is chosen orthogonal
to the channel vectors of the $k$ \textit{nearest} interferers
$\bh_1, \ldots, \bh_k$.
The parameter $k$ must be an integer and satisfy $k \leq \Nr - 1$.\footnote{Although performance could
conceivably be improved by choosing $k$ on the
basis of the channel realizations, it is sufficient for our purposes
to choose $k$ in an offline fashion. The value of $k$ is left
unspecified for the time being.}  Amongst the
filters satisfying the orthogonality requirement $|\bv_0^{\dagger}
\bh_i |^2 = 0$ for $i=1, \ldots, k$, we are interested in the one
that maximizes the desired signal power $|\bv_0^{\dagger} \bh_0
|^2$.  By simple geometry, this corresponds to choosing $\bv_0$ in
the direction of the projection of vector $\bh_0$ on the nullspace
of vectors $(\bh_1,\ldots, \bh_k)$.
More precisely, if the columns of the $\Nr \times (\Nr - k)$ matrix
${\bf Q}$ form an orthonormal basis for the nullspace of
$(\bh_1,\ldots, \bh_k)$,
then the receive filter is chosen as:
\begin{eqnarray}
\bv_0 = \frac{ {\bf Q}^{\dagger} \bh_0}{ || {\bf Q}^{\dagger} \bh_0
|| }.
\end{eqnarray}
The corresponding outage probability and maximum density are denoted
$\Pout^{\textrm{pzf}-k}(\lambda)$
and $\lambda_{\epsilon}^{\textrm{pzf}-k}(\Nr)$, respectively.
Notice that $k=0$ and $k= \Nr - 1$ correspond to the extremes of
MRC (${\bf v}_0 = {\bf h}_0 /  ||{\bf h}_0||$)
and interference cancellation of the maximum number of interferers (full zero forcing), respectively.
  Because
the MMSE receiver is SINR-maximizing, we clearly have
\begin{equation} \label{eq-mmse_pzf}
\Pout^{\textrm{mmse}}(\lambda) \leq \Pout^{\textrm{pzf}-k}(\lambda) ~~~~ \textrm{and}
~~~ \lambda_{\epsilon}^{\textrm{mmse}}(\Nr) \geq \lambda_{\epsilon}^{\textrm{pzf}-k}(\Nr)
\end{equation}
for all $k$ and any set of system parameters.

Although suboptimal, it is beneficial to study the PZF receiver because it is more amenable to analysis than the MMSE filter
 and because its simple
structure allows us to clearly understand why linear  scaling is achievable.
Note, however, that an MMSE filter should be used in practice because it also is linear and
its CSI requirements are less stringent (PZF requires knowledge of individual interferer channels,
whereas MMSE only requires knowledge of the aggregate interference).

\section{Main Results: Density Scaling with Receive Antennas}
\label{sec:main}

In this section we prove the main result of the paper, which is that
$\lambda_{\epsilon}^{\textrm{mmse}}(\Nr)$ and $\lambda_{\epsilon}^{\textrm{pzf}-k}(\Nr)$ both increase
linearly with $\Nr$. We prove this result in two parts: we first show that a lower bound
to $\lambda_{\epsilon}^{\textrm{pzf}-k}(\Nr)$, and thus to
$\lambda_{\epsilon}^{\textrm{mmse}}(\Nr)$, increases linearly with $\Nr$, and then
show that an upper bound on $\lambda_{\epsilon}^{\textrm{mmse}}(\Nr)$ also
increases linearly with $\Nr$.

\subsection{Lower Bound: Achievability of Linear Scaling with Partial Zero Forcing}
\label{sec:pzf}

First, we show that linear scaling is achievable by finding a lower bound on $\lambda_{\epsilon}^{\textrm{pzf}-k}(\Nr)$ that
is linear in $\Nr$.  In order to develop the bound, we first statistically characterize the signal and
interference coefficients when the PZF receiver is used.
These characterizations are a consequence of the basic result that the squared-norm of the projection
of a $\Nr$-dimensional vector with iid unit-variance complex Gaussian components onto an
independent $s$-dimensional subspace is $\chi^2_{2s}$ \cite{Muirhead}.\footnote{Throughout the
paper we abide by communications literature convention and define a $\chi^2_{2s}$ random
variable to have PDF $f(x) = \frac{x^{s-1}e^{-x} }{(s-1)!}$, which may differ slightly
from the definition in probability literature.}

We denote the signal and interference coefficients as
\begin{eqnarray} \label{eq-coeff1}
S &\triangleq& |\bv_0^{\dagger} \bh_0 |^2 \\
H_i &\triangleq& |\bv_0^{\dagger} \bh_i |^2 ~~~~ i= 1,2, \ldots
\label{eq-coeff2}
\end{eqnarray}
and characterize the statistics of these coefficients in the
following lemma.
\begin{lemma} \label{lemma_statistics}
For PZF-$k$,  the signal coefficient $S$ is $\chi^2_{2(\Nr - k)}$,
the interference terms $H_1,\ldots, H_k$ are zero, and coefficients
$H_{k+1}, H_{k+2}, \ldots$ are iid unit-mean exponential (i.e.,
$\chi^2_2$). Furthermore, $S, H_{k+1}, H_{k+2}, \ldots$ are mutually
independent.
\end{lemma}
\begin{proof}
See Appendix \ref{app:stats}.
\end{proof}
\vspace{2mm}

Using this statistical characterization and the definitions in
(\ref{eq-coeff1})-(\ref{eq-coeff2}) the received SINR is
\begin{eqnarray} \label{eq:ibra}
\sinr^{\textrm{pzf}-k} = \frac{ S } {\frac{1}{\snr} + d^{\alpha} \sum_{i =
k+1}^{\infty} |X_i|^{-\alpha} H_i }
\end{eqnarray}
where the $S$ and $H_i$ terms are characterized in Lemma
\ref{lemma_statistics}, the quantities $|X_{k+1}|^2, |X_{k+2}|^2,
\ldots$ are the $k+1, k+2, \ldots$ ordered points of a 1-D PPP with
intensity $\pi \lambda$, and all random variables are independent.
The aggregate interference power for PZF-$k$ is denoted as:
\begin{eqnarray} \label{eq:byeeto}
I_k \triangleq d^{\alpha} \sum_{i = k+1}^{\infty} |X_i|^{-\alpha} H_i .
\end{eqnarray}
and the expectation of this interference power is characterized in the following lemma:
\begin{lemma} \label{lemma_expected_int}
For $k > \frac{\alpha}{2} - 1$, the expected interference power is
characterized as:
\begin{align}
\E[I_k ]&= \left( \pi d^2 \lambda \right)^{\frac{\alpha}{2}} \sum_{i
= k+1}^{\infty} \frac{ \Gamma
\left(i- \frac{\alpha}{2} \right)  } {\Gamma(i)}
< \left( \pi d^2 \lambda \right)^{\frac{\alpha}{2}} \left(\frac{\alpha}{2} - 1\right)^{-1}
\left( k - \left\lceil \frac{\alpha}{2} \right\rceil
\right)^{1-\frac{\alpha}{2}},
\end{align}
where $\Gamma(\cdot)$ is the gamma function and $\lceil \cdot \rceil$
is the ceiling function, with the upper bound valid for $k > \left\lceil \frac{\alpha}{2}
\right\rceil$.
\end{lemma}
\begin{proof}
See Appendix \ref{proof:expected}.
\end{proof}
\vspace{2mm}

To derive the main result for PZF, we use Lemma \ref{lemma_expected_int} to upper bound
$\E[1/\sinr^{\textrm{pzf}-k}]$ and then combine this with
Markov's inequality to reach an outage probability upper bound, and, inversely, a density lower bound:
\begin{theorem} \label{thm_bound}
The outage probability with PZF-$k$ is upper bounded by:
\begin{equation}
\Pout^{\textrm{pzf}-k} (\lambda) \leq \frac{ \beta \left(
\left( \pi d^2 \lambda \right)^{\frac{\alpha}{2}} \left(\frac{\alpha}{2} - 1\right)^{-1} \left( k
- \left\lceil \frac{\alpha}{2} \right\rceil
\right)^{1-\frac{\alpha}{2}}  + \frac{1}{\snr} \right) } {\Nr - k -
1}
\end{equation}
for $\left\lceil \frac{\alpha}{2} \right\rceil < k < \Nr - 1$.
In turn, the maximum density $\lambda_{\epsilon}^{\textrm{pzf}-k}(\Nr)$ is lower bounded by:
\begin{equation}
\lambda_{\epsilon}^{\textrm{pzf}-k}(\Nr)  \geq    \left( \frac{\epsilon}{\beta}
\right)^{\frac{2}{\alpha}}  \frac{\left(\frac{\alpha}{2} - 1\right)^{\frac{2}{\alpha}} }{\pi
d^2 } \left(\Nr - k - 1 - \frac{\beta}{\epsilon ~ \snr}
\right)^{\frac{2}{\alpha}} \left( k - \left\lceil \frac{\alpha}{2}
\right\rceil \right)^{1-\frac{2}{\alpha}}
\end{equation}
for any $k$ satisfying $\left\lceil \frac{\alpha}{2} \right\rceil < k < \Nr - 1 - \frac{\beta}{\epsilon ~ \snr}$.
\end{theorem}
\begin{proof}
The outage upper bound is derived by rewriting the
outage probability as the tail probability of random variable
$1/\sinr$ and then applying Markov's inequality as follows:
\begin{eqnarray}
\Pout^{\textrm{PZF}-k} (\lambda)
&=& \prob \left[ \frac{1}{\sinr^{\textrm{pzf}-k}} \geq \frac{1}{\beta} \right]  \\
&\stackrel{(a)}{\leq}&
\beta  \cdot \E \left[ \frac{1}{\sinr^{\textrm{pzf}-k}}  \right]\\
&\stackrel{(b)}{=}&  \beta \cdot \E \left[ I_k + \frac{1}{\snr}\right] \E \left[ \frac{1}{S} \right] \\
&\stackrel{(c)}{<}&  \frac{ \beta \left( \left( \pi d^2
\lambda \right)^{\frac{\alpha}{2}} \left(\frac{\alpha}{2} - 1\right)^{-1} \left( k - \left\lceil
\frac{\alpha}{2} \right\rceil \right)^{1-\frac{\alpha}{2}}  +
\frac{1}{\snr} \right) } {\Nr - k - 1},
\end{eqnarray}
where (a) is due to Markov's inequality, (b) is due to (\ref{eq:ibra}) the
independence of $I_k$ and $S$, and (c) follows from Lemma \ref{lemma_expected_int}
and because $S$ is $\chi^2_{2(\Nr - k)}$ and $\E[1/\chi^2_{2l} ] = 1/(l-1)$ for $l > 1$.
Setting this bound equal to $\epsilon$
and then solving for $\lambda$ yields the associated lower bound to
$\lambda_{\epsilon}$.
\end{proof}
\vspace{2mm}

It is worthwhile to note that the $\left(\Nr - k - 1 -
\frac{\beta}{\epsilon ~ \snr} \right)^{\frac{2}{\alpha}}$ term in
the $\lambda_{\epsilon}$ lower bound is the density increase due to
 array gain (i.e., increased signal power), while the $\left( k -
\left\lceil \frac{\alpha}{2} \right\rceil
\right)^{1-\frac{2}{\alpha}}$ term is the density increase due to
interference cancellation.  Thus, the bound succintly illustrates the tradeoff between array gain and interference cancellation.

In order to show the achievability of linear scaling, we need only appropriately increase $k$ with $\Nr$.
If we choose the number of cancelled interferers $k =\theta \Nr$ for some constant
$0 < \theta < 1$, the density lower bound becomes:
\begin{equation}
\lambda_{\epsilon}^{\textrm{pzf}-\theta \Nr}(\Nr)   \geq    \left(
\frac{\epsilon}{\beta} \right)^{\frac{2}{\alpha}}
\frac{\left(\frac{\alpha}{2} - 1\right)^{\frac{2}{\alpha}} }{\pi d^2 } (1 -
\theta)^{\frac{2}{\alpha}} \theta^{1-\frac{2}{\alpha}}
 \left( \Nr - \frac{1 + \frac{\beta}{\epsilon ~ \snr}}{1-\theta}
\right)^{\frac{2}{\alpha}} \left(\Nr - \theta^{-1} \left\lceil
\frac{\alpha}{2} \right\rceil \right)^{1-\frac{2}{\alpha}}.
\label{eq-density_lower}
\end{equation}

Because the conditions for Theorem \ref{thm_bound} are satisfied for sufficiently large $\Nr$ if $k=\theta \Nr$ with $0 < \theta < 1$
(for any $\epsilon >0$ and $\snr > \beta > 0$), the lower bound scales linearly with $\Nr$.  This result is formally stated as follows:
\begin{lemma} \label{lemma_lower_linear}
For any $\theta$ satisfying $0 < \theta < 1$,
\begin{equation}
\frac{\lambda_{\epsilon}^{\textrm{pzf}-\theta \Nr}(\Nr)}{\Nr} \geq \left(
\frac{\epsilon}{\beta} \right)^{\frac{2}{\alpha}}
\frac{\left(\frac{\alpha}{2} - 1\right)^{\frac{2}{\alpha}} }{\pi d^2 } (1 -
\theta)^{\frac{2}{\alpha}} \theta^{1-\frac{2}{\alpha}},
\end{equation}
for sufficiently large $\Nr$.
\end{lemma}

This perhaps surprising scaling result can be intuitively understood by examining how
the signal and aggregate interference power increase with $\Nr$.
Choosing $\theta < 1$ ensures that the signal power,
which is $\chi^2_{2(1-\theta) \Nr}$, increases linearly with $\Nr$.
Based on the upper bound in Lemma \ref{lemma_expected_int} we can
see that the condition $\theta
> 0$ ensures that the interference power increases only linearly
with $\Nr$ if $\lambda$ is linear in $\Nr$.  These linear terms are
offsetting, and thus allow an approximately constant SINR to be
maintained as $\lambda$ is increased linearly with $\Nr$.

On the other hand, linear scaling does not occur if $k$ is kept constant.
Specifically, if $k=\kappa$ for some constant $\kappa$,  then the signal power increases linearly with $\Nr$ as
desired but the interference power increases too quickly with the density (as $\lambda^{\alpha/2}$),
thereby limiting the density growth to $\Nr^{2/\alpha}$. Linear scaling also does not hold at the other
extreme where all but a fixed number of degrees of freedom are
used for cancellation: if $k = \Nr - \kappa$ for some constant $\kappa$, then
the interference power scales appropriately with the density but the signal power is $\chi^2_{2
\kappa}$ and thus does not increase with $\Nr$, thereby limiting the
density increase to $\Nr^{1-2/\alpha}$.  These results are consistent with the findings of
\cite{HunAnd08} and \cite{HuaAndSub}.

\subsection{Upper Bound: The MMSE Receiver}
\label{sec:mmse}

While the earlier result showed that a lower bound to $\lambda_{\epsilon}^{\textrm{mmse}}(\Nr)$ scales linearly
with $\Nr$, we make our scaling characterization more precise by finding an upper bound
to $\lambda_{\epsilon}^{\textrm{mmse}}(\Nr)$ that also scales linearly with $\Nr$.
In order to obtain such a bound,
we utilize the MMSE performance outage probability lower bound from \cite{Gao_Smith}.
Extended to the model in this paper, the bound states:
\begin{equation} \label{mmse_bound}
\Pout^{\textrm{mmse}} (\lambda) \geq \prob \left[ \frac{ d^{-\alpha} || {\bf h_0} ||^2 }{ \sum_{i=\Nr}^{\infty}  |X_i|^{-\alpha} H_i} \leq \beta \right]
\end{equation}
where ${\bf h}_0$ and $|X_1|, |X_2|, \ldots$ are defined as before, and
the random variables $H_1, H_2, \ldots$ are iid, unit-mean exponential random variables (independent of all other
random variables).  The bound  in \cite{Gao_Smith} is for fixed interferer distances and no thermal
noise.  However, by averaging over the interferer locations (according to the PPP) and using the fact that outage
probability is increasing in the noise power $\eta$, we obtain (\ref{mmse_bound}) and see that it also holds in the presence of noise.

The SIR expression in (\ref{mmse_bound}) is closely related to the SINR characterization for PZF-$k$  in (\ref{eq:ibra}).
The denominator of the SIR in (\ref{mmse_bound}) is precisely as if a PZF receiver with
$k=\Nr-1$ is used (i.e., the nearest $\Nr-1$ interferers are cancelled, and the effective fading coefficients from
the uncancelled interferers are iid exponential), while the numerator corresponds to PZF with $k=0$.
Thus, the bound in (\ref{mmse_bound}) corresponds
to an idealized setting where the receive filter cancels the nearest $\Nr - 1$ interferers but still
is in the direction of ${\bf h}_0$.

We translate this outage lower bound into a density upper bound in a manner that
is complementary to Theorem \ref{thm_bound}: we upper bound the expected SIR
and then apply Markov's inequaltiy to the success probability to obtain the following result:
\begin{theorem} \label{thm_upper_bound}
The outage probability with an MMSE receiver is lower bounded by:
\begin{eqnarray}
\Pout^{\textrm{mmse}} (\lambda)\geq
1 - \frac{d^{-\alpha}}{\beta} \left( \frac{ 2 \Nr + 1 + \frac{\alpha}{2} }{\pi \lambda} \right)^{\alpha/2}
\end{eqnarray}
and, in turn,  $\lambda_{\epsilon}^{\textrm{mmse}}(\Nr)$ is upper bounded by:
\begin{eqnarray}
\lambda_{\epsilon}^{\textrm{mmse}}(\Nr) \leq
\frac{ 2 \Nr + 1 + \frac{\alpha}{2} }{ \pi d^2 \beta^{2/\alpha} (1-\epsilon)^{2/\alpha}}.
\end{eqnarray}
\end{theorem}
\begin{proof}
See Appendix \ref{sec:proof-UB}.
\end{proof}

This upper bound, which by (\ref{eq-mmse_pzf}) also applies to PZF, scales linearly with $\Nr$.
Combining Theorems 1 and 2, it is clear that $\lambda_{\epsilon}^{\textrm{mmse}}(\Nr)$ and $\lambda_{\epsilon}^{\textrm{pzf}}(\Nr)$ both scale
linearly with $\Nr$, although numerical results in \ref{sec:discussion} will confirm that MMSE is better by a non-negligible constant factor.

Because the SINR with a PZF receiver differs only slightly from the expression in (\ref{mmse_bound}), we can use precisely
the argument of Theorem \ref{thm_upper_bound} to derive the following upper bound on $\lambda_{\epsilon}^{\textrm{pzf}-k}(\Nr)$ that,
unlike Theorem \ref{thm_bound}, applies for any $0 \leq k \leq \Nr - 1$:
\begin{equation} \label{pzf_density_upper}
\lambda_{\epsilon}^{\textrm{pzf}-k}(\Nr) \leq
\frac{k + l + \alpha/2}{ \pi d^2 \beta^{2/\alpha} (1-\epsilon)^{2/\alpha}}
\left( \frac{\Nr -k}{l-1} \right)^{2/\alpha}   ~~~~ \forall l > 1.
\end{equation}
This allows us to upper bound the performance if MRC or full zero-forcing is used.
For $k=0$ (MRC), if we choose $l=2$ the upper bound becomes
\begin{eqnarray}
\lambda_{\epsilon}^{\textrm{pzf}}(0) \leq
\frac{2 + \alpha/2}{ \pi d^2 \beta^{2/\alpha} (1-\epsilon)^{2/\alpha}} \Nr^{2/\alpha},
\end{eqnarray}
while for $k=\Nr-1$ (full zero-forcing) we choose $l=\Nr + 1$ to get
\begin{eqnarray}
\lambda_{\epsilon}^{\textrm{pzf}}(\Nr-1) \leq
\frac{2 + \alpha/(2 \Nr)}{ \pi d^2 \beta^{2/\alpha} (1-\epsilon)^{2/\alpha}} \Nr^{1-2/\alpha}.
\end{eqnarray}
For MRC the upper bound is $O(\Nr^{2/\alpha})$ while for full zero-forcing it is $O(\Nr^{1-2/\alpha})$.  These
upper bounds complement the matching lower bounds in \cite{HunAnd08} and \cite{HuaAndSub}, respectively.

\subsection{Array Gain v. Interference Cancellation}

Because the MMSE receiver implicitly balances, through (\ref{eq-mmse_filter}), array gain and interference cancellation,
it is not evident how the MMSE utilizes the receive degrees of freedom.
If the eigenvalues of the interference covariance $\Sigma$ are roughly equal then the MMSE filter is nearly
in the direction of ${\bf h}_0$; on the other hand, if the eigenvalues are very disparate then the MMSE filter is
(approximately) in the direction of the projection of ${\bf h}_0$ on the subspace orthogonal to the directions of the
strong interfering eigenmodes.  Thus, the fraction of degrees of freedom used for
array gain instead of interference cancellation depends critically on
the spread of the eigenvalues of $\Sigma$, which turns out to depend on the path loss exponent $\alpha$.

The MMSE receiver can be understood by studying the PZF receiver, and in particular
by finding the value of $\theta$ that maximizes the PZF density.
Based on Lemma \ref{lemma_lower_linear} it is clear that the PZF
density lower bound depends on $\theta$ only through the term $(1 - \theta)^{\frac{2}{\alpha}}
\theta^{1-\frac{2}{\alpha}}$ for large $\Nr$.
To determine the dependence of the PZF upper bound in (\ref{pzf_density_upper}),
we first minimize the bound with respect to $l$, for large $Nr$ and $k=\theta \Nr$.  A simple calculation finds that
$l = \left( \frac{2/\alpha}{1 - 2/\alpha} \right) \theta \Nr$ is the minimizer, and with this choice of $l$ the
dependence of the density upper bound also occurs only through the term $(1 - \theta)^{\frac{2}{\alpha}}
\theta^{1-\frac{2}{\alpha}}$.

Thus, the upper and lower bounds depend on $\theta$ only through the term $(1 - \theta)^{\frac{2}{\alpha}}
\theta^{1-\frac{2}{\alpha}}$.  By taking the derivative (w.r.t. $\theta$) and solving, we find that the
maximizing value of $\theta$ is:
\begin{eqnarray}
\theta^* = 1 - \frac{2}{\alpha}.
\end{eqnarray}
As $\alpha \to 2$ the degrees of freedom should be used to
boost signal power rather than to cancel interference (i.e.,
$\theta^* \rightarrow 0$), because far-away interference is significant
and so cancelling a few nearby interferers provides a smaller benefit than
using the antennas for array gain.  At the other extreme,  $\theta^* \rightarrow 1$ as the path loss exponent increases
because the power from nearby interferers begins to dominate and thus the antennas are more profitably used for interference cancellation than array gain.

As it turns out, the optimizing value $\theta^*$ and the above intuition are also consistent with the MMSE receiver.
Fig. \ref{fig:eig} contains a plot of $\E \left[ |\bv_0^{\dagger} \bh_0 |^2 / (||\bh_0 ||^2 ||\bv_0 ||^2) \right]$, the expectation
of the squared correlation between the normalized MMSE filter and channel vector, versus $\alpha$ for $\Nr=8$.
This quantity effectively measures the fraction of degrees of freedom used for array gain.
For the PZF receiver this metric is precisely equal to $\Nr - k$, and thus would be equal to $1-\theta^* = 2/\alpha$
if $k$ was chosen as $k=\theta^* \Nr$.  Although the MMSE receiver implicitly balances array gain and cancellation
(as opposed to the explicit balance for the PZF receiver), the plot shows that the MMSE receiver also utilizes
(approximately) a fraction $2/\alpha$ of its receive degrees of freedom for array gain.  Similar to the intuition
stated for the behavior of $\theta^*$ for PZF, the eigenvalues of $\Sigma$ become more disparate
as $\alpha$ is increased, and the MMSE receiver takes advantage of this by performing more interference cancellation (and thus
providing less array gain) when $\alpha$ is larger.

\begin{figure}
\centering
\includegraphics[width=4.5in]{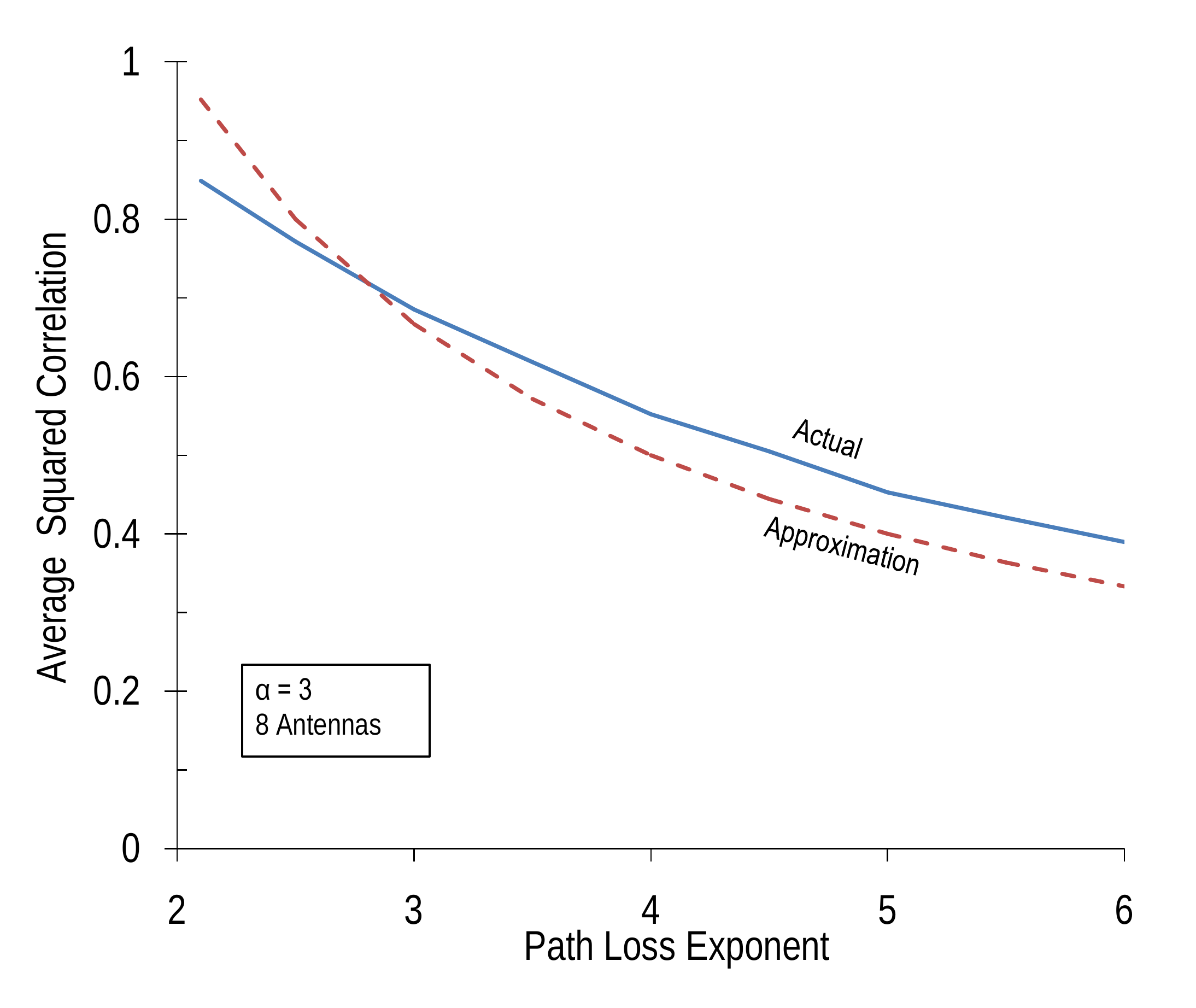}
\caption{Average squared correlation between signal channel ${\bf h}_0$ and
MMSE filter as a function of path loss exponent $\alpha$.  Also shown is the
approximation $2/\alpha$.}
\label{fig:eig}
\end{figure}

\subsection{Improved Lower and Upper Bounds}

Although the bounds developed in Sec. \ref{sec:pzf} and \ref{sec:mmse} are sufficient to show linear scaling,
both of them are quite loose.  This looseness is primarily due to the use of Markov's inequality, and the following
more accurate bounds are instead derived through application of Chebychev's inequality:
\begin{theorem} \label{thm:pzf-cheby}
The outage probability for the PZF filter is upper bounded by
\begin{equation}
P_{\rm out}^{\mathrm{pzf-}k}(\lambda) \leq \prob \left( S \leq \sigma^* \right) + \beta^2 \mathrm{Var}^{\rm ub}(I_k) \int_{\sigma^*}^{\infty} \frac{1}{\left(s - \frac{\beta}{\mathrm{SNR}} - \beta \E[I_k] \right)^2} f_S(s) d s.
\end{equation}
where $S \sim \chi^2_{2(\Nr-k)}$, $I_k$ is defined in (\ref{eq:byeeto}),
 $\sigma^* = \beta \left(\E[I_k] + \frac{1}{\mathrm{SNR}} + \sqrt{\mathrm{Var}^{\rm ub}(I_k)}\right)$,
and
\begin{equation*}
\mathrm{Var}^{\rm ub}(I_k) = (\pi d^2 \lambda)^{\alpha} \left( \sum_{i=k+1}^{\infty} \left(\E[T_i^{-\alpha} ] + \mathrm{Var}(T_i^{-\frac{\alpha}{2}}) \right) + 2 \sum_{i=k+1}^{\infty} \sqrt{\mathrm{Var}\left(T_i^{-\frac{\alpha}{2}}\right)} \sum_{j=i+1}^{\infty} \sqrt{\mathrm{Var}\left(T_j^{-\frac{\alpha}{2}}\right) } \right). \end{equation*}
\end{theorem}
\begin{proof}
See Appendix \ref{sec:proof-cheby1}.
\end{proof}
\vspace{2mm}

\begin{theorem} \label{thm:mmse-cheby}
The outage probability for the MMSE filter is lower bounded by
\begin{equation}
P_{\rm out}^{\mathrm{mmse}}(\lambda) \geq 1-\frac{\mathrm{Var}^{\mathrm{ub}}\left(S/I_{\Nr}\right)}{\left(\beta - \E^{\mathrm{lb}}\left[ S/I_{\Nr}\right]\right)^2}
\end{equation}
where $S \sim \chi^2_{2\Nr}$, $I_{\Nr}$ is defined in (\ref{eq:byeeto}), and
\begin{eqnarray}
\E^{\mathrm{lb}}\left[ S/I_{\Nr}\right] &=&
\Nr/(\pi d^2 \lambda)^{\frac{\alpha}{2}} \sum_{i=\Nr}^{\infty} \E[T_i^{-\frac{\alpha}{2}}] \nonumber \\
\mathrm{Var}^{\mathrm{ub}}\left(S/I_{\Nr}\right) & = & \frac{\Nr(\Nr+1)}{(\pi d^2 \lambda)^{\alpha} \left( \sum_{i=\Nr}^{\infty} e^{-\left(\gamma+\frac{\alpha}{2} \psi_0(i) \right)} \right)^2 } - \frac{\Nr^2}{(\pi d^2 \lambda)^{\alpha} \left(\sum_{i=\Nr}^{\infty} \frac{\Gamma \left(i-\frac{\alpha}{2}\right)}{\Gamma(i)} \right)^2},
\end{eqnarray}
and $\gamma$ is the Euler-Mascheroni constant and $\psi_0(i)$ is the poly-Gamma function.
\end{theorem}
\begin{proof}
See Appendix \ref{sec:proof-cheby2}.
\end{proof}
\vspace{2mm}

In the two theorems the random variables $T_i$ are each chi-square with $2i$ degrees of freedom, and have moments characterized by:
\begin{equation}
\E[T_i^b] = \frac{\Gamma(i+b)}{\Gamma(i)}, ~~
\mathrm{Var}(T_i^b) = \frac{\Gamma(i+2b)}{\Gamma(i)} - \left( \frac{\Gamma(i+b)}{\Gamma(i)} \right)^2, ~ \mbox{ if } i+b > 0.
\end{equation}
By equating these bounds to $\epsilon$ and solving (numerically) for $\lambda$,
the PZF and MMSE densities can be lower and upper bounded, respectively.

\subsection{Numerical Results}
\label{sec:discussion}

In Figures \ref{fig:alpha3} and \ref{fig:alpha4} the numerically computed maximum densities for
the MMSE receiver $\lambda_{\epsilon}^{\textrm{mmse}}(\Nr)$ and the PZF receiver $\lambda_{\epsilon}^{\textrm{pzf-k}}(\Nr)$ with
$k=\theta^* \Nr$ are plotted on a log-log scale versus $\Nr$ for $\alpha=3$ and $\alpha=4$,
along with the PZF lower bounds (from Theorems \ref{thm_bound} and \ref{thm:pzf-cheby}), the MMSE upper bounds
(from Theorems \ref{thm_upper_bound} and \ref{thm:mmse-cheby}), and
the densities for MRC and full zero-forcing ($\lambda_{\epsilon}^{\textrm{pzf-k}}(\Nr)$ with $k=0$ and $k=\Nr-1$, respectively).
In each plot, the tighter of the upper and the tighter of the lower bounds correspond to the Chebychev-based
bounds in the previous section.  The bounds and numerically computed densities are representative of
the linearly increasing density for PZF and MMSE, whereas MRC and full zero forcing both exhibit much poorer scaling.
Figures \ref{fig:alpha3_linear} and \ref{fig:alpha4_linear} provide linear plots of the maximum density versus $\Nr$ for more realistic
numbers of antennas.  Even a few antennas allow for very large density gains, and thus the asymptotic scaling results  also are
indicative of performance for small values of $\Nr$.  The plots also make it clear that MMSE and PZF are strongly preferred
to MRC or full zero forcing, and also that a non-negligible benefit is afforded by using the optimal MMSE filter rather than PZF.

\begin{figure}
\centering
\includegraphics[width=4.5in]{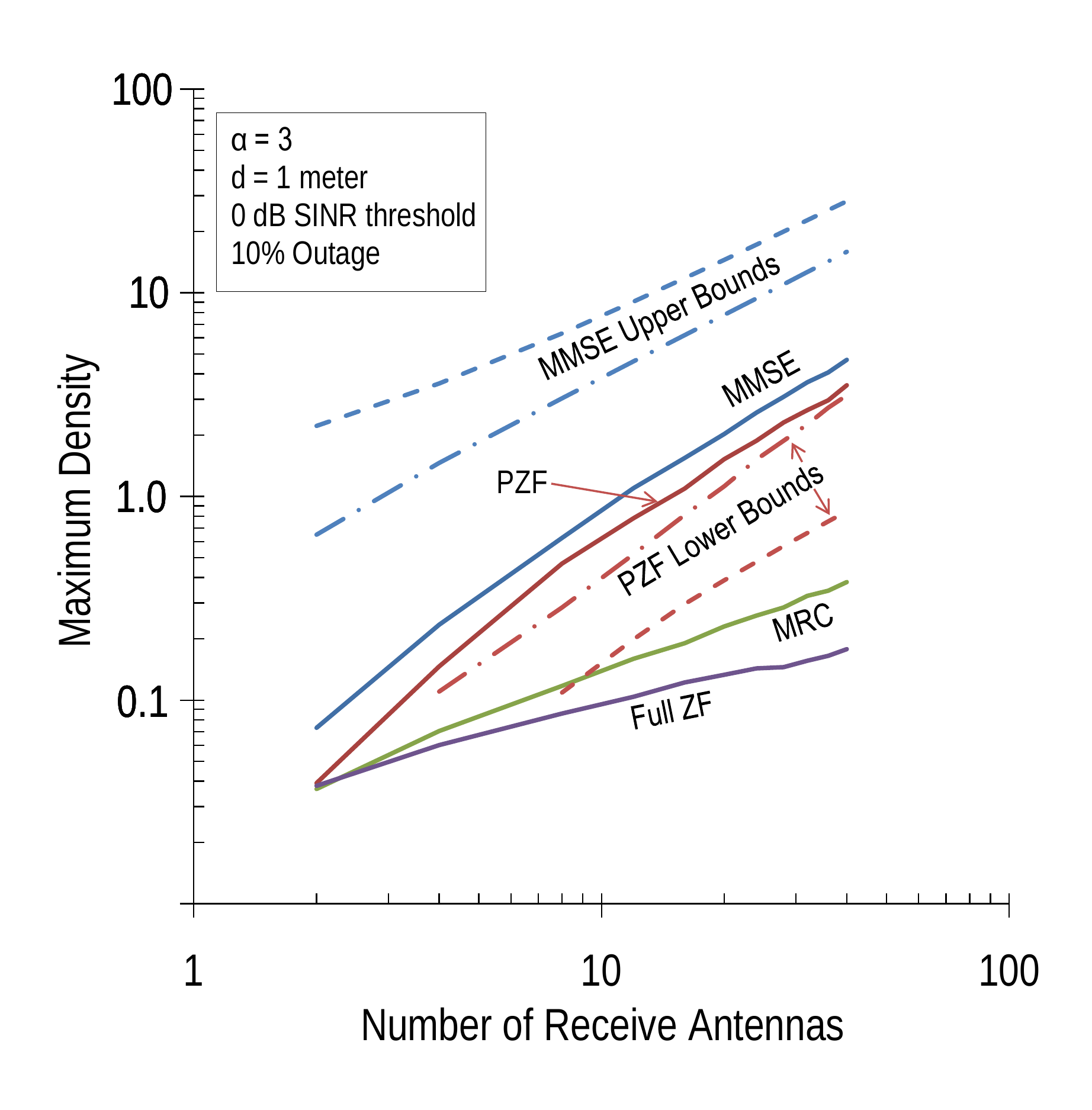}
\caption{Maximum density versus $\Nr$ for $\epsilon=.1$, $\beta=1$,
$\alpha=3$, $\theta = \frac{1}{3}$, $d=1$.} \label{fig:alpha3}
\end{figure}

\begin{figure}
\centering
\includegraphics[width=4.5in]{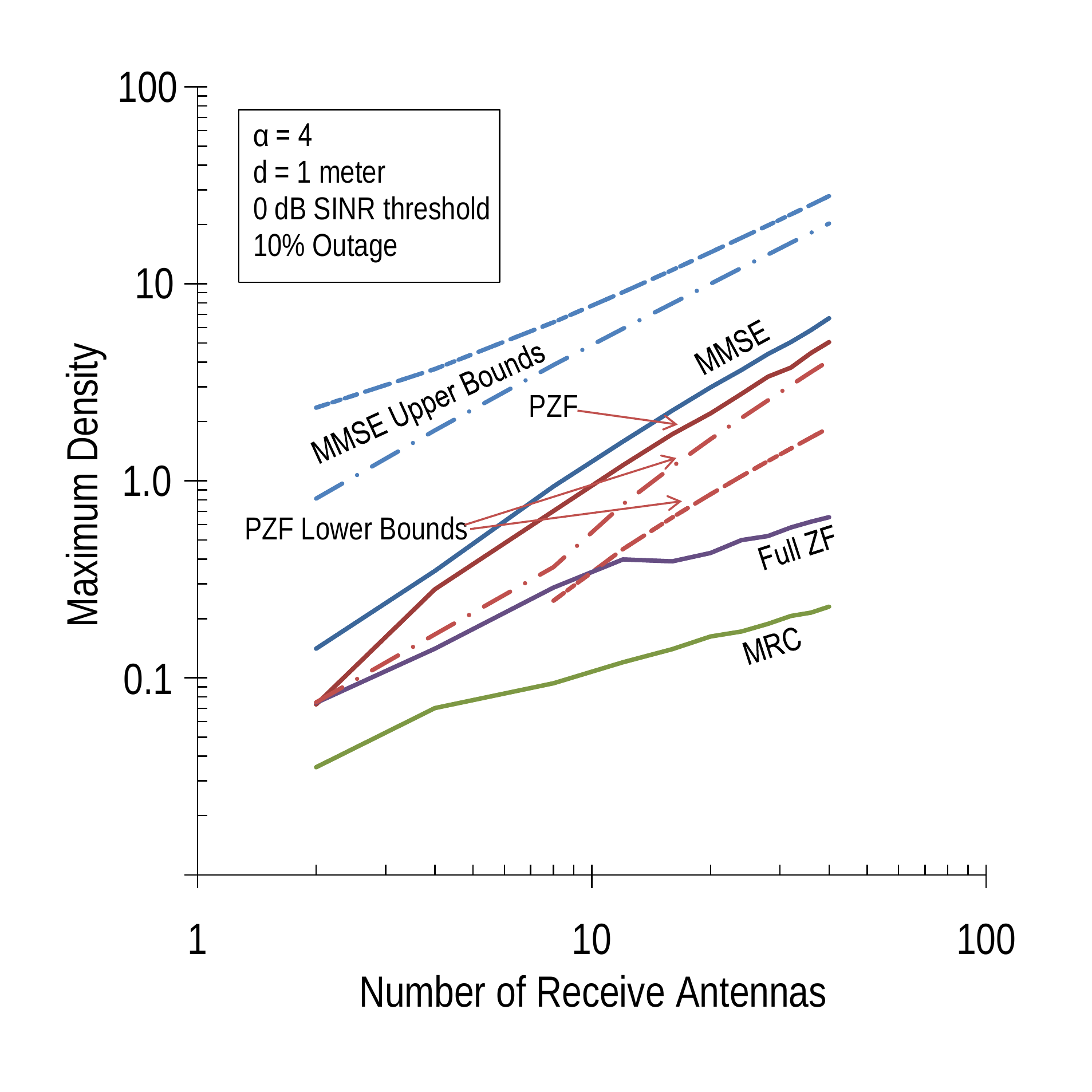}
\caption{Maximum density versus $\Nr$ for $\epsilon=.1$, $\beta=1$,
$\alpha=4$, $\theta = \frac{1}{2}$, $d=1$.} \label{fig:alpha4}
\end{figure}

\begin{figure}
\centering
\includegraphics[width=4.5in]{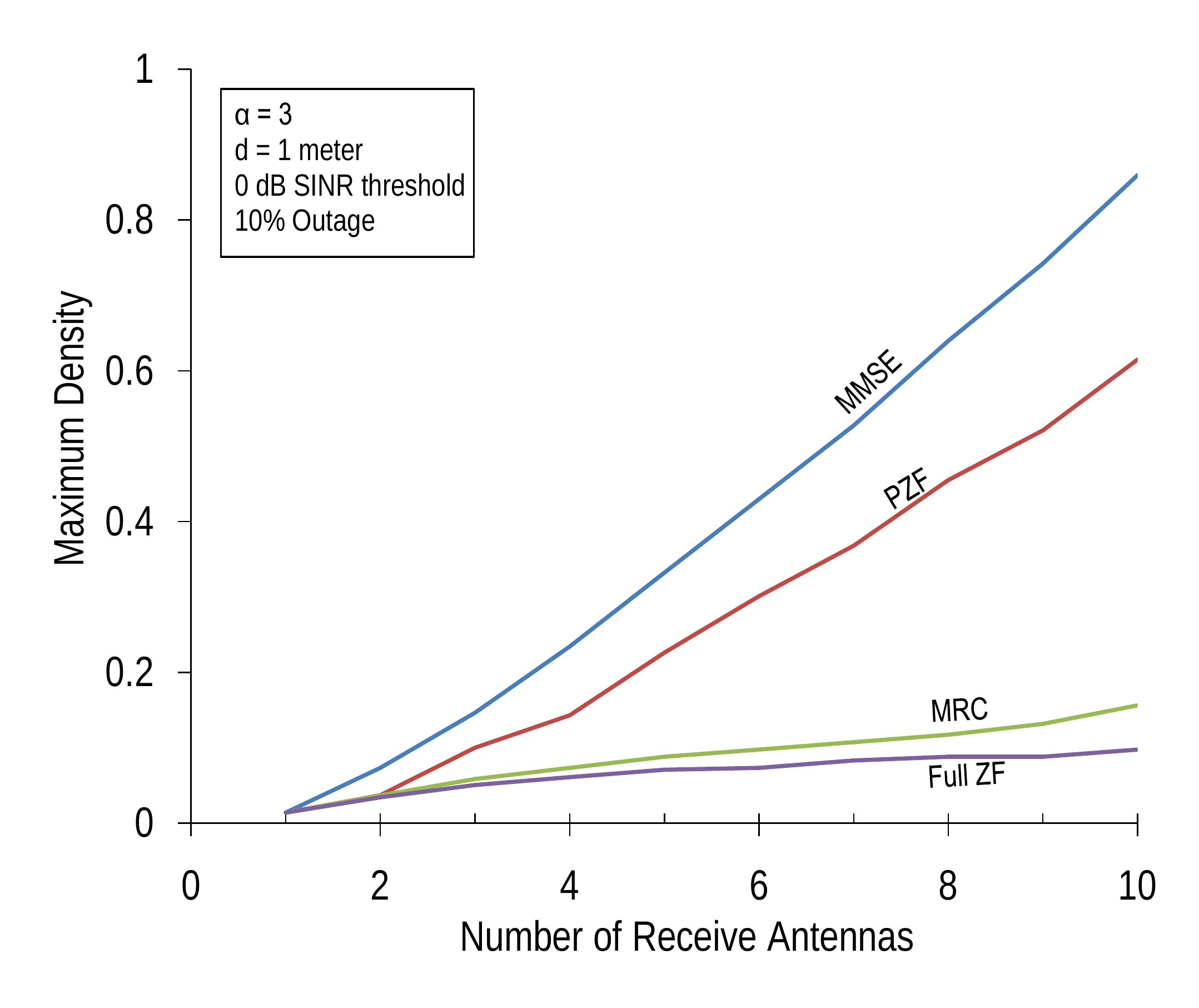}
\caption{Maximum density versus $\Nr$ for $\epsilon=.1$, $\beta=1$,
$\alpha=3$, $\theta = \frac{1}{3}$, $d=1$.} \label{fig:alpha3_linear}
\end{figure}

\begin{figure}
\centering
\includegraphics[width=4.5in]{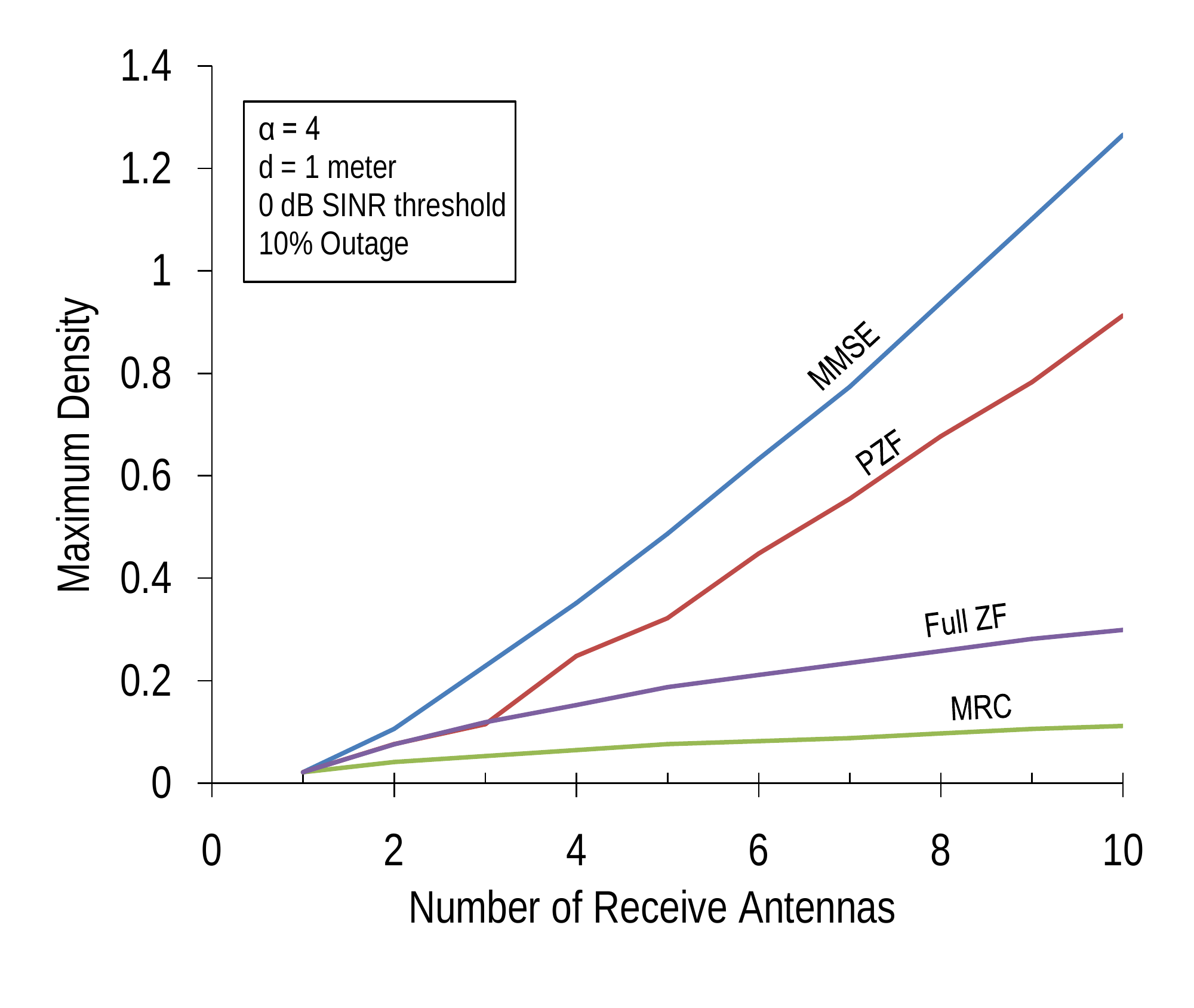}
\caption{Maximum density versus $\Nr$ for $\epsilon=.1$, $\beta=1$,
$\alpha=4$, $\theta = \frac{1}{2}$, $d=1$.} \label{fig:alpha4_linear}
\end{figure}

\section{Generalizations and Extensions of the Model}
\label{sec:extensions}

In this section, we explore three of the potentially controversial aspects of our model to show that the linear scaling result is not an artifact of
our model and assumptions: we remove the assumption of perfect CSI at the receiver, we evaluate the benefit of antennas from an
end-to-end perspective, and we consider the importance of the interferer geometry.

\subsection{Effect of Imperfect CSI}
\label{sec:csi}

The objective of this section is to illustrate that reasonable performance is achieved even if the receiver has to estimate the
CSI, instead of assuming this information is \textit{a priori} provided to the receiver.  In order to design the optimal MMSE filter,
the receiver requires an estimate of ${\bf h}_0$, the signal channel,  and ${\boldsymbol \Sigma}$, the interference (plus noise)
covariance. The desired channel  can be estimated at the receiver via pilot symbols and the effects
of such training error are well understood \cite{HasHoc03}.
On the other hand, it is not as clear how the receiver can estimate ${\boldsymbol \Sigma}$ and what
effect estimation error has on performance.

Recall that the covariance ${\boldsymbol \Sigma}$ depends on the
interferer locations, and thus on the active interferers,
as well as the instantaneous channel realizations.  Although coordinated transmission of pilots
seems infeasible in a decentralized network, the receiver can estimate the covariance by listening to interferer
transmissions, in the absence of desired signal.  If the desired transmitter remains quiet for $K$ symbols,
the receiver can use the $K$ observations of noise plus interference to form the sample covariance
\begin{equation}
\hat{{\boldsymbol \Sigma}} \triangleq \frac{1}{K} \sum_{i=1}^K {\bf r}_i {\bf r}_i^\dagger
\end{equation}
where ${\bf r}_i$ represents the $i$-th observation of the noise plus interference.  Assuming, for simplicity, knowledge
of ${\bf h}_0$, the receiver can then use the filter $\hat{{\boldsymbol \Sigma}}^{-1} {\bf h}_0$ and the corresponding
SINR is
\begin{equation}
\sinr = \frac{ \left( \bh_0^\dagger \hat{{\boldsymbol \Sigma}}^{-1} \bh_0 \right) ^2}
{\bh_0^\dagger \hat{{\boldsymbol \Sigma}}^{-\dagger}
{\boldsymbol \Sigma} \hat{{\boldsymbol \Sigma}}^{-1} {\bf h}_0}.
\end{equation}
This SINR was analyzed in \cite{ReedMallet74} assuming that all interferers transmit independent Gaussian symbols, and it was shown that for every ${\boldsymbol \Sigma}$,
the expected SINR using filter $\hat{{\boldsymbol \Sigma}}^{-1} {\bf h}_0$ (expectation w.r.t
the distribution of $\hat{{\boldsymbol \Sigma}}$), and the SINR using the correct filter
${\boldsymbol \Sigma}^{-1} {\bf h}_0$ are related according to:
\begin{equation}
\E_{\hat{{\boldsymbol \Sigma}}}\left[ \frac{ \left( \bh_0^\dagger \hat{{\boldsymbol \Sigma}}^{-1} \bh_0 \right) ^2}
{\bh_0^\dagger \hat{{\boldsymbol \Sigma}}^{-\dagger}
{\boldsymbol \Sigma} \hat{{\boldsymbol \Sigma}}^{-1} {\bf h}_0} \right] = \left(1-\frac{\Nr -1}{K+1}\right)
{\bf h}_0^{\dagger} {\boldsymbol \Sigma}^{-1}  {\bf h}_0 ,
\end{equation}
where from (\ref{sinr-mmse}),  ${\bf h}_0^{\dagger} {\boldsymbol \Sigma}^{-1}  {\bf h}_0$ is the SINR when the proper MMSE filter is used.  By taking an additional expectation with respect to ${\boldsymbol \Sigma}$ (which is determined by the interferer
locations and channels), we  see that the expected SINR using an MMSE filter based upon the sample covariance
$\hat{{\boldsymbol \Sigma}}$ is precisely a factor of $1-\frac{\Nr -1}{K+1}$ smaller than the expected SINR with perfect
knowledge of ${\boldsymbol \Sigma}$.
As expected, this factor is increasing in $K$ and converges to one as $K \rightarrow \infty$, because $\hat{{\boldsymbol \Sigma}} \rightarrow {\boldsymbol \Sigma}$ as $K \rightarrow \infty$.  If $K=2\Nr - 3$, the expected SINR is decreased by $3$ dB.

Although the result of \cite{ReedMallet74} applies to the expected SINR, numerical results confirm that the
results also apply to the outage scenario considered here.  Therefore, a system using the
sample covariance from $K$ observations and SINR threshold $\beta$ has, approximately, the same maximum density as a system
with perfect CSI and SINR threshold $\beta / \left(1-\frac{\Nr -1}{K+1}\right)$.  From the various bounds, we see that the
maximum density depends on the SINR threshold as $\beta^{-2/\alpha}$.  Thus, using the $K$-observation sample covariances
instead of the true covariance reduces the density, approximately, by a factor of $\left(1-\frac{\Nr -1}{K+1}\right)^{2/\alpha}$.
If we choose $K=2\Nr - 3$ then the loss factor is $2^{2/\alpha}$ for all $\Nr$; therefore, by appropriately scaling $K$
\textit{linearly} with $\Nr$  performance within a constant factor of the perfect CSI benchmark is achieved.

To solidify these conclusions, in Fig. \ref{fig:csi} the maximum density is plotted versus $K$ for a $6$ antenna system with $\snr=10$ dB.
The curves correspond to perfect knowledge of ${\bf h}_0$, estimation of ${\bf h}_0$ on the basis of two interference-free pilots
(each at $10$ dB), and the approximation of the perfect CSI density ($0.41$ in this case) multiplied by
$\left(1-\frac{\Nr -1}{K+1}\right)^{2/\alpha}$.   From the figure we see that estimation of ${\bf h}_0$ does not significantly
reduce density, the approximate density expression is reasonably accurate, and that choosing $K$ on the order of $10$ or $20$
leads to a density reasonable close to the perfect CSI benchmark.  Indeed, even if such estimation must be performed for
every transmission,
the overhead is reasonable in light of the fact that packets are typically on the order of hundreds of symbols.

\begin{figure}
\centering
\includegraphics[width=4.5in]{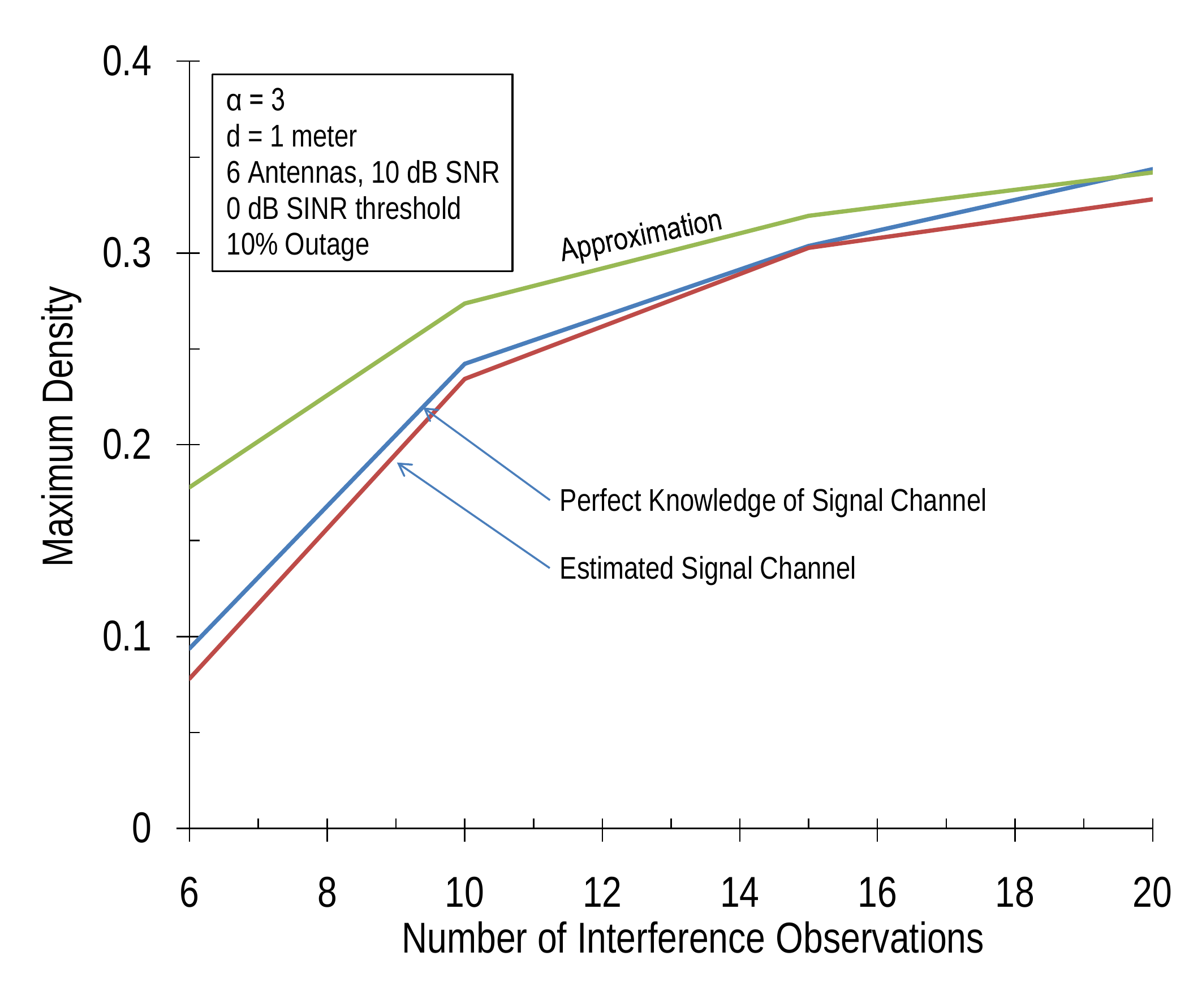}
\caption{Maximum density versus $K$, number of interference observations, for $\epsilon=.1$, $\beta=1$,
$\alpha=3$, $d=1$, and $\snr=10$ dB.} \label{fig:csi}
\end{figure}

\subsection{End-to-End Throughput}
\label{sec:efp}

The transmission capacity quantifies per-hop performance, whereas end-to-end throughput depends on the range and rate of each
transmission and the spatial intensity of such transmissions.
Thus, a legitimate question is whether the linear scaling of transmission capacity translates into linear scaling
of end-to-end throughput?  In the context of the transport capacity \cite{GupKum00}, which is a widely accepted end-to-end metric, this question
can be answered in the affirmative.  Transport capacity is the product of rate and distance summed over all
transmissions  and thus is proportional to $\lambda_{\epsilon} (1-\epsilon) d \log_2(1+\beta)$,
i.e., the product of the successful transmission density, per-hop distance, and per-hop rate.  Since the transport capacity
is linearly proportional to $\lambda_{\epsilon}$, the linear density scaling established in Theorems 1 and 2 also translates
into linear scaling of the transport capacity.

Although increasing the density with $\Nr$ leads to linear scaling of the transport capacity, it is not \textit{a priori}
clear if the antennas should instead be used to increase the transmission rate and/or range.
Based on the scaling
results in Theorems 1 and 2, we see that $\lambda_{\epsilon} d^2 \beta^{2/\alpha} \propto \Nr$.
Thus, if the density and rate (i.e., SINR threshold $\beta$) are kept constant, then the range can be increased at order
$d \propto \sqrt{\Nr}$.  Alternatively the SINR threshold can be increased at order
$\beta \propto \Nr^{\frac{\alpha}{2}}$, which translates to increasing per-hop rate approximately as
$\frac{\alpha}{2} \log_2(1+\Nr)$.  Because transport capacity is proportional to $\lambda_{\epsilon} d \log_2(1+\beta)$,
using the receive antennas to increase per-hop range only increases transport capacity at order $\sqrt{\Nr}$
while increasing per-hop rate leads to an even poorer logarithimic increase (consistent with \cite{GovBli07}).
Therefore, the most efficient use of the receive
array, from an end-to-end perspective, is to increase the density of simultaneous transmissions rather than the per-transmission
rate or distance.

These points can be argued concretely within the framework of \emph{expected forward progress} (EFP), a metric introduced in \cite{Bac06} that is defined as
\begin{equation}
{\rm EFP} = \nu p \cdot \E[X_0] \cdot \log_2(1+\beta),
\end{equation}
where $\lambda = \nu p$ is the density of transmitters subject to an ALOHA protocol where the entirety of nodes in the network (that have density $\nu$)
transmit with probability $p$ and act as receivers (and hence relays) with probability $1-p$, where $p$ is a design parameter that is
optimized offline.  An opportunistic routing protocol is employed through which the successful receiving node (i.e., relay) offering the most geographic progress towards a defined destination direction is selected to forward each transmitter's packet, and the quantity $\E[X_0]$ is the expected progress
offered by such relay.  The expected distance is inversely proportional to the transmitter density $p$, and by understanding how the
optimum $p$ changes with respect to $\Nr$ we can determine if receive antennas are more effectively used for
increasing transmit density ($p$ increasing rapidly with $\Nr$) or for increasing per-hop range ($p$ approximately constant).

Fig. \ref{fig:EFP} contains plots of EFP versus transmission probability $p$ for $\Nr=1$ to $\Nr=8$ and
the optimum value of $p$ is seen to increase approximately
linearly with $\Nr$, confirming that it is more effective to use the antennas to increase density.
(Closer examination shows that the expected distance per communication, i.e. the per-hop range, is approximately
constant with respect to $\Nr$ for the optimizing value $p$.)
The plot also illustrates that the EFP itself is linear with $\Nr$, confirming that linear scaling
holds for multihop wireless networks.  To see that increasing the rate rather than density also is suboptimal,
if $p$ is kept fixed to the small value of $0.075$ when $\Nr=8$ and the spectral efficiency
$\log_2(1+\beta)$ is set to $4$ (this value leads to the same expected per-hop range as the optimizing $p$),
then the resulting EFP is $0.3$ instead of the $0.4$ achievable if rate is kept fixed and $p$ (i.e. the density)
is increased.

\begin{figure}
\centering
\includegraphics[width=4in]{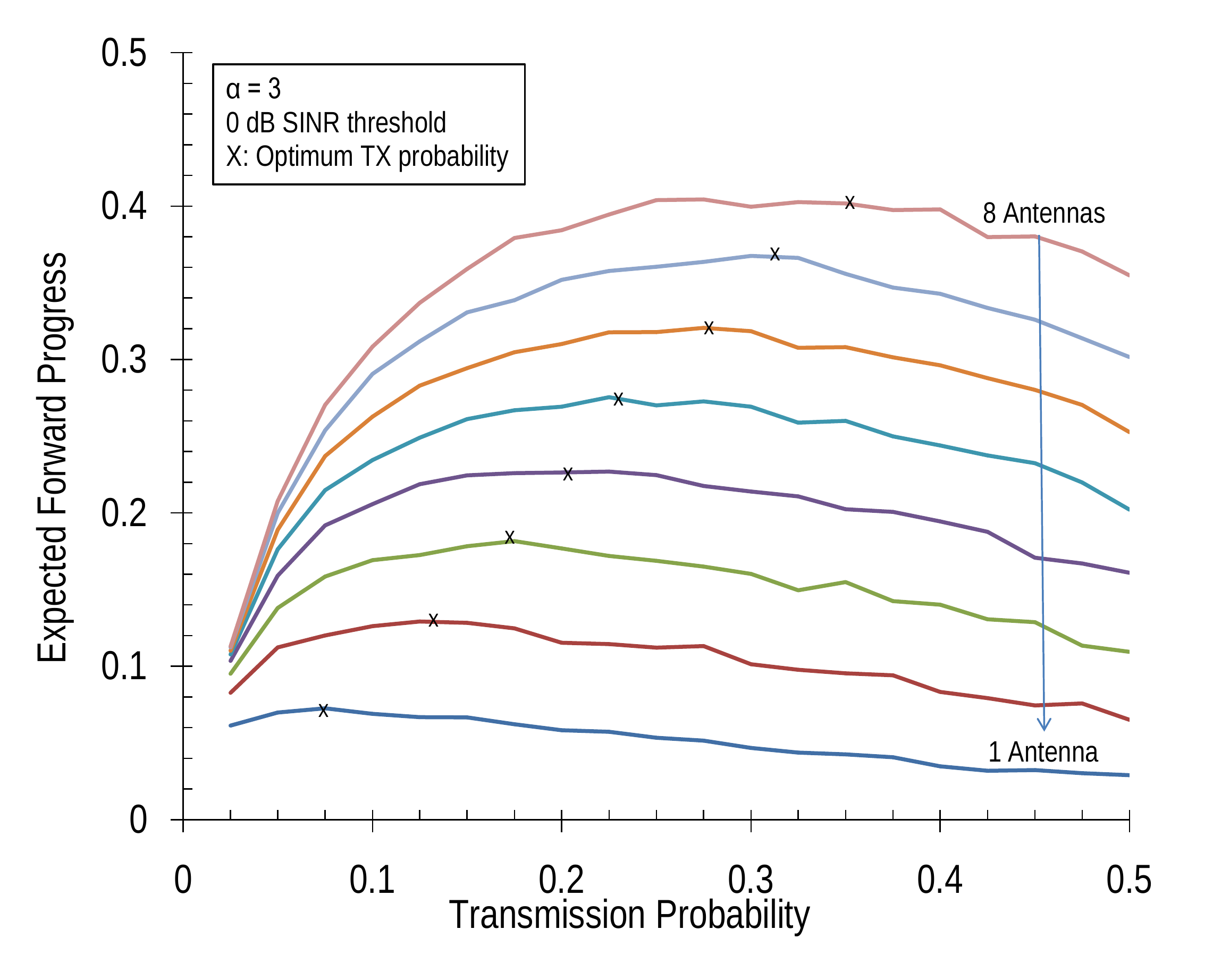}
\caption{Expected forward progress vs. transmission probability ($p$) for $\Nr$ ranging from $1$ (bottom) to $8$ (top), for $\alpha=3$,
$\beta = 0$ dB, and no noise.  The optimizing values of $p$ are denoted with an X.}
\label{fig:EFP}
\end{figure}

\subsection{Effect of Interference Geometry} \label{sec:geometry}

Finally, we consider the effect of the interference geometry.  In this paper we have assumed a homogeneous Poisson distribution for the
node locations, which is a realistic model if the users take up random locations and do not coordinate their transmissions.  One might
reasonably wonder, however, if the linear scaling result is an artifact of this model, since in a Poisson field the nearest interferers
dominate and so interference cancellation might be far more profitable in this setup than in a more regular network.  A ``good'' MAC
protocol would seemingly space out the active transmitters at any instance, to avoid dominant interference.  As a simple manifestation
of such a MAC, we consider a regular network where the interferers take up positions on a square grid with edges of
length $1/\sqrt{\lambda}$. From Fig. \ref{fig:RegVsPoisson} we see that a regular network allows for a larger density of simultaneous transmissions,
but only by a constant factor that is independent of $\Nr$.  Based on this, we conjecture that the linear scaling result holds for
any reasonable network geometry.

\begin{figure}
\centering
\includegraphics[width=4.5in]{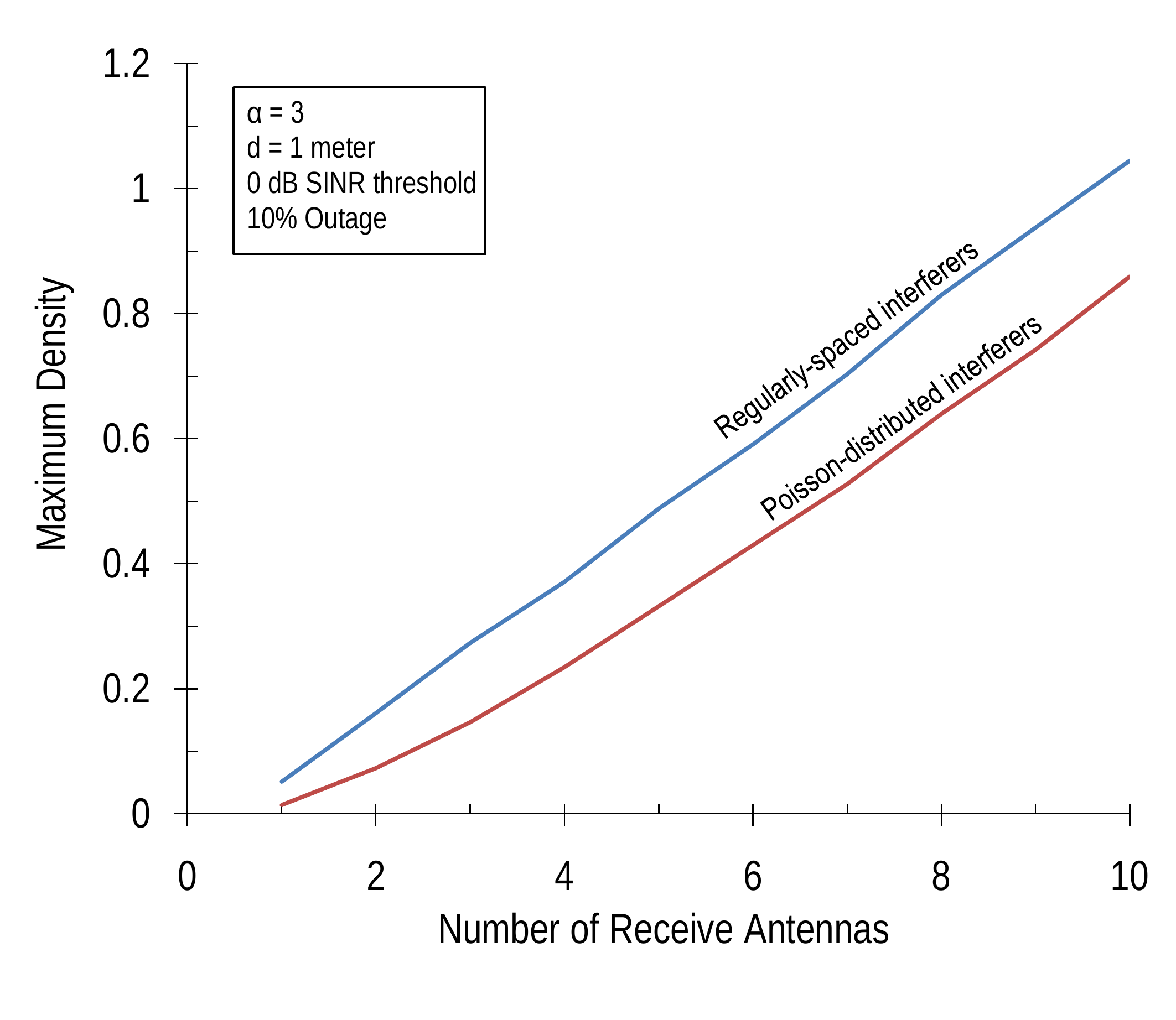}
\caption{The optimal density for both Regular and Poisson networks increases linearly with $\Nr$, which regular networks having a slightly higher achievable density since nearby (dominant) interferers do not exist.}
\label{fig:RegVsPoisson}
\end{figure}

\section{Conclusion}
\label{sec:conc}

The main takeaway of this paper is that very large throughput gains can be achieved
in ad hoc networks using only receive antennas in conjunction with  linear processing.
In a point-to-point link receive antennas only provide array gain, which translates into a
linear SNR and thus logarithimic rate increase (in the number of receive antennas).  In an ad hoc network, however, receive antennas
can also be used to cancel interference and this possibility turns out to yield
much more significant benefits.  In particular,  the main result of the paper showed that
using receive antennas to cancel interference and obtain some array gain allows the density of simultaneous transmissions
to be increased linearly with the number of receive antennas when nodes transmit using only a single antenna.
This result only requires channel state information at the receiver,
which can be reasonably estimated, and the conclusion was seen to be robust to the particular interferer
geometry.

From an end-to-end perspective, this linear increase in the density of simultaneous transmissions
naturally translates into a linear increase in network-wide throughput.
In addition, our analysis showed that receive antennas are in fact best utilized
by increasing the density of simultaneous transmissions rather than increasing
the per-hop rate or range.
Finally, although the single-transmit/multiple-receive antenna setting  may appear artificial, subsequent work on this model has shown that it is actually
detrimental to employ multiple transmit antennas when channel state information is not available
to the transmitter \cite{VazHea09}.  As a result, the single transmit/multiple receive antenna
setting is indeed very relevant.


\appendices

\section{Proof of Lemma \ref{lemma_statistics}} \label{app:stats}

By the definition of $\bv_0$, the quantity $|\bv_0^{\dagger}
\bh_0|^2$ is the squared-norm of the projection of vector $\bh_0$ on
Null($\bh_1,\ldots, \bh_k$).  This nullspace is $\Nr - k$
dimensional (with probability one) by basic properties of iid
Gaussian vectors and is independent of $\bh_0$ by the independence
of the channel vectors, and thus $|\bv_0^{\dagger} \bh_0|^2$ is $\chi^2_{2(\Nr - k)}$ \cite{Muirhead}.
The second property holds by the definition of the PZF-$k$ receiver. To
prove the third property, note that $\bv_0$ depends only on $\bh_0,
\bh_1, \ldots, \bh_k$ and thus is independent of $\bh_{k+1},
\bh_{k+2}, \ldots$.  Because the distribution of each channel vector
is rotationally invariant (i.e., the distributions of ${\bf W}
\bh_i$ and $\bh_i$ are the same for any unitary matrix ${\bf W}$),
we can perform a change of basis such that $\bv_0 = [1 ~ 0 ~ \cdots
0]^T$.  After this change of basis, each $\bv_0^{\dagger} \bh_i$
(for $i \geq k+1$) is simply equal to the first component of
$\bh_i$.  As a result $\bv_0^{\dagger} \bh_{k+1}, \bv_0^{\dagger}
\bh_{k+2}, \ldots $ are iid complex Gaussians; thus the squared
norms are iid exponentials, and furthermore these terms are
independent of $S$.

\section{Proof of Lemma \ref{lemma_expected_int}} \label{proof:expected}

First we have
\begin{eqnarray}
\E \left[ \sum_{i=k+1}^{\infty} |X_i|^{-\alpha} H_i \right]
&=&  \sum_{i=k+1}^{\infty} \E \left[ |X_i|^{-\alpha} H_i \right]
=  \sum_{i=k+1}^{\infty} \E \left[ |X_i|^{-\alpha} \right]
\end{eqnarray}
where $\E \left[ |X_i|^{-\alpha} H_i \right] = \E \left[
|X_i|^{-\alpha} \right] \E \left[ H_i \right] = \E \left[
|X_i|^{-\alpha} \right] $ due to the independence of $|X_i|$ and
$H_i$ and the fact that $\E[H_i]=1$.
Because $|X_1|^2, |X_2|^2, \ldots$ are
a 1-D PPP with intensity $\pi \lambda$,  random variable $\pi
\lambda |X_i|^2$ is $\chi^2_{2i}$ and thus has PDF $f(x) =
\frac{x^{i-1}e^{-x} }{(i-1)!}$.  Therefore
\begin{eqnarray}
\E \left[ \left(|X_i|^2 \right)^{-\alpha/2} \right]
&=& \left( \pi \lambda \right)^{\frac{\alpha}{2}} \int_0^{\infty}
x^{-\alpha/2} \frac{x^{i-1}e^{-x} }{(i-1)!} dx \\
&=& \left( \pi \lambda \right)^{\frac{\alpha}{2}} \frac{ \Gamma
\left(i- \frac{\alpha}{2} \right)  } {\Gamma(i)}.
\end{eqnarray}
This quantity is finite only for $i > \frac{\alpha}{2}$, and thus
the expected power from the nearest uncancelled interferer is finite
only if $k +1 > \frac{\alpha}{2}$.

To reach the upper bound, we use the following inequality from \cite{HuaAndSub}
(which is derived using Kershaw's inequality to the gamma function):
\begin{eqnarray}
\frac{ \Gamma \left(i- \frac{\alpha}{2} \right)  } {\Gamma(i)} <
\left(i - \left\lceil \frac{\alpha}{2} \right\rceil
\right)^{-\frac{\alpha}{2}}
\end{eqnarray}
where $\lceil {\cdot} \rceil$ is the ceiling function and we require
$i > \left\lceil \frac{\alpha}{2} \right\rceil$. Therefore
\begin{eqnarray}
\sum_{i = k+1}^{\infty} \frac{ \Gamma \left(i- \frac{\alpha}{2}
\right)  } {\Gamma(i)} &<& \sum_{i = k+1}^{\infty} \left(i -
\left\lceil \frac{\alpha}{2} \right\rceil
\right)^{-\frac{\alpha}{2}} \\
 &\leq&  \int_k^{\infty} \left(x - \left\lceil
\frac{\alpha}{2} \right\rceil
\right)^{-\frac{\alpha}{2}} dx \\
&=& \left(\frac{\alpha}{2} - 1 \right)^{-1} \left( k - \left\lceil
\frac{\alpha}{2} \right\rceil \right)^{1-\frac{\alpha}{2}},
\end{eqnarray}
where the inequality in the second line holds because
$x^{-\frac{\alpha}{2}}$ is a decreasing function.


\section{Proof of Theorem \ref{thm_upper_bound}}
\label{sec:proof-UB}

The outage upper bound is obtained by keeping the interference contribution of only the nearest $l$
uncancelled interferers in (\ref{mmse_bound}) and applying Markov's inequality to the \textit{success} probability:
\begin{align}
1 - \Pout^{\textrm{mmse}} (\lambda) &\stackrel{(a)}{\leq}  \prob \left[ \frac{ d^{-\alpha} || {\bf h_0} ||^2 }{ \sum_{i=\Nr}^{\infty}  |X_i|^{-\alpha} H_i} \geq \beta \right] \\
&\stackrel{(b)}{\leq} \prob \left[ \frac{ d^{-\alpha} || {\bf h_0} ||^2 }{ \sum_{i=\Nr}^{\Nr -1 + l}  |X_i|^{-\alpha} H_i} \geq \beta \right] \\
&\stackrel{(c)}{\leq} \prob \left[ \frac{ d^{-\alpha} || {\bf h_0} ||^2 }{|X_{\Nr -1 + l}|^{-\alpha} \sum_{i=\Nr}^{\Nr -1 + l}   H_i} \geq \beta \right] \\
&\stackrel{(d)}{\leq}
\frac{1}{\beta} ~ \E \left[ \frac{ d^{-\alpha} || {\bf h_0} ||^2 }{ |X_{\Nr  - 1 + l}|^{-\alpha}
 \sum_{i=\Nr}^{\Nr -1 + l}   H_i} \right] \label{eq:success}
\end{align}
where (a) follows from (\ref{mmse_bound}), (b)  because decreasing the interference increases the SIR and thus the success probability,
(c)  because $|X_i|$ are increasing in $i$ and the function $(\cdot)^{-\alpha}$ is decreasing, and (d) is due to Markov's inequality.
By the independence of the various random variables:
\begin{align}
 \E \left[ \frac{ d^{-\alpha} || {\bf h_0} ||^2 }{ |X_{\Nr - 1 + l}|^{-\alpha}
 \sum_{i=\Nr}^{\Nr + l - 1}   H_i} \right]
 &=  \mathbb{E} \left[|| {\bf h_0} ||^2 \right] \mathbb{E} \left[ \frac{d^{-\alpha}}{\sum_{i=\Nr}^{\Nr -1 + l}   H_i} \right]
  \mathbb{E} \left[|X_\Nr -1 + l|^{\alpha} \right]  \\
   &=   \frac{\Nr}{l-1} \left( \pi d^2 \lambda \right)^{-\alpha/2}
   \frac{ \Gamma \left(\Nr -1+ l + \frac{\alpha}{2} \right) }{\Gamma \left(\Nr -1+ l \right)}. \label{expect}
\end{align}
where we have used the fact that $|| {\bf h_0} ||^2 \sim \chi^2_{2\Nr}$,
$ \sum_{i=\Nr}^{\Nr + l - 1} H_i$ is the sum of $l$ iid exponentials and thus is $\chi^2_{2l}$, and
$|X_\Nr - 1 + l|^2 \sim \frac{1}{\pi \lambda} \chi^2_{2(\Nr -1 + l)}$.

By applying Kershaw's inequality, which states
$\Gamma \left(x+1 \right) / \Gamma \left(x+s \right) < \left(x - \frac{1}{2} + \sqrt{s + 1/4} \right)^{1-s}$
$\forall x > 0$ and $0 < s < 1$, and the property $\Gamma(x+1)=x\Gamma(x)$, we have:
\begin{equation} \label{kershaw2}
   \frac{ \Gamma \left(\Nr -1 + l + \frac{\alpha}{2} \right) }{\Gamma \left(\Nr -1 + l \right)}
\leq \left( \Nr -1 + l + \frac{\alpha}{2}\right)^{\alpha/2}.
\end{equation}
Substituting (\ref{expect}) and (\ref{kershaw2}) into (\ref{eq:success}) yields:
\begin{align}
\Pout^{\textrm{mmse}}(\lambda) &\geq
1 - \frac{1}{\beta} \frac{\Nr}{l-1} \left( \pi \lambda \right)^{-\alpha/2}
\left( \Nr -1 + l + \frac{\alpha}{2}\right)^{\alpha/2}
\end{align}
By choosing $l = \Nr + 1$ we obtain the desired outage probability lower bound, and by setting this bound
to $\epsilon$ and solving we get the density upper bound.

\section{Proof of Theorem \ref{thm:pzf-cheby}}
\label{sec:proof-cheby1}
Write $I$ for $I_k^{\mathrm{pzf-}k}$.  The variance of a sum of rvs may be expressed as the sum of the covariances:
\begin{eqnarray}
\mathrm{Var}(I) &=& (\pi d^2 \lambda)^{\alpha} \sum_{i=k+1}^{\infty} \sum_{j=k+1}^{\infty} \mathrm{Cov}\left(T_i^{-\frac{\alpha}{2}} H_i,T_j^{-\frac{\alpha}{2}} H_j\right) \nonumber \\
&=& (\pi d^2 \lambda)^{\alpha}\sum_{i=k+1}^{\infty} \sum_{j=k+1}^{\infty} \E[(T_i T_j)^{-\frac{\alpha}{2}} H_i H_j] - \E[T_i^{-\frac{\alpha}{2}} H_i] \E[T_j^{-\frac{\alpha}{2}} H_j] \nonumber \\
&=& (\pi d^2 \lambda)^{\alpha} \left( \sum_{i=k+1}^{\infty} \left(2 \E[T_i^{-\alpha} ] - \E[T_i^{-\frac{\alpha}{2}}]^2 \right) + 2 \sum_{i=k+1}^{\infty} \sum_{j=i+1}^{\infty}  \left(\E[(T_i T_j)^{-\frac{\alpha}{2}}] - \E[T_i^{-\frac{\alpha}{2}}] \E[T_j^{-\frac{\alpha}{2}}] \right) \right) \nonumber \\
&=& (\pi d^2 \lambda)^{\alpha} \left( \sum_{i=k+1}^{\infty} \left(\E[T_i^{-\alpha} ] + \mathrm{Var}(T_i^{-\frac{\alpha}{2}}) \right) + 2 \sum_{i=k+1}^{\infty} \sum_{j=i+1}^{\infty}  \mathrm{Cov}\left(T_i^{-\frac{\alpha}{2}},T_j^{-\frac{\alpha}{2}} \right) \right).
\end{eqnarray}
Recall the Cauchy-Scharz inequality applied to the covariance of two rvs gives
\begin{equation}
\mathrm{Cov}(X,Y) \leq \sqrt{\mathrm{Var}(X)\mathrm{Var}(Y)}.
\end{equation}
Applied here:
\begin{eqnarray}
\mathrm{Var}(I) & \leq & (\pi d^2 \lambda)^{\alpha} \left( \sum_{i=k+1}^{\infty} \left(\E[T_i^{-\alpha} ] + \mathrm{Var}(T_i^{-\frac{\alpha}{2}}) \right) + 2 \sum_{i=k+1}^{\infty} \sum_{j=i+1}^{\infty}  \sqrt{\mathrm{Var}\left(T_i^{-\frac{\alpha}{2}}\right)\mathrm{Var}\left(T_j^{-\frac{\alpha}{2}}\right) } \right) \nonumber \\
& = & (\pi d^2 \lambda)^{\alpha} \left( \sum_{i=k+1}^{\infty} \left(\E[T_i^{-\alpha} ] + \mathrm{Var}(T_i^{-\frac{\alpha}{2}}) \right) + 2 \sum_{i=k+1}^{\infty} \sqrt{\mathrm{Var}\left(T_i^{-\frac{\alpha}{2}}\right)} \sum_{j=i+1}^{\infty} \sqrt{\mathrm{Var}\left(T_j^{-\frac{\alpha}{2}}\right) } \right).
\end{eqnarray}
We apply the Chebychev inequality as follows.  Condition on all possible values for $S$, fix a parameter $\sigma \geq \beta \left(\E[I] + \frac{1}{\mathrm{SNR}}\right)$, and upper bound the tail probability on $I$ by one for all $s < \sigma$, then apply the Chebychev bound on $I-\E[I]$ for each $s \geq \sigma$:
\begin{eqnarray}
P_{\rm out}^{\mathrm{pzf-}k} &=& \prob(\mathrm{SINR}^{\mathrm{pzf-}k} \leq \beta) \nonumber \\
&=& \int_0^{\infty} \prob \left( I \geq \frac{s}{\beta} -\frac{1}{\mathrm{SNR}} \right) f_S(s) d s \nonumber \\
&=& \int_0^{\sigma} \prob \left( I \geq \frac{s}{\beta} -\frac{1}{\mathrm{SNR}} \right) f_S(s) d s + \int_{\sigma}^{\infty} \prob \left( I \geq \frac{s}{\beta} -\frac{1}{\mathrm{SNR}} \right) f_S(s) d s \nonumber \\
& \leq & \prob \left( S \leq \sigma \right) + \int_{\sigma}^{\infty} \prob \left( I \geq \frac{s}{\beta} -\frac{1}{\mathrm{SNR}} \right) f_S(s) d s \nonumber \\
& = & \prob \left( S \leq \sigma \right) + \int_{\sigma}^{\infty} \prob \left( I-\E[I] \geq \frac{s}{\beta} -\frac{1}{\mathrm{SNR}} - \E[I] \right) f_S(s) d s \nonumber \\
& \leq & \prob \left( S \leq \sigma \right) + \int_{\sigma}^{\infty} \prob \left( \left|I-\E[I] \right| \geq \frac{s}{\beta} -\frac{1}{\mathrm{SNR}} - \E[I] \right) f_S(s) d s \nonumber \\
& \leq & \prob \left( S \leq \sigma \right) + \int_{\sigma}^{\infty} \frac{\mathrm{Var}(I)}{\left(\frac{s}{\beta} -\frac{1}{\mathrm{SNR}} - \E[I] \right)^2} f_S(s) d s \nonumber \\
& = & \prob \left( S \leq \sigma \right) + \beta^2 \mathrm{Var}(I) \int_{\sigma}^{\infty} \frac{1}{\left(s - \frac{\beta}{\mathrm{SNR}} - \beta \E[I] \right)^2} f_S(s) d s
\end{eqnarray}
Viewed as a function of $\sigma$, this outage probability upper bound is a convex function in $\sigma$ with derivative:
\begin{equation}
\frac{d }{d \sigma} P_{\rm out}^{\mathrm{pzf-}k,\mathrm{ub}}(\sigma) = f_S(\sigma) - \beta^2 \mathrm{Var}(I) \frac{f_S(\sigma)}{(\sigma - \frac{\beta}{\mathrm{SNR}} - \beta \E[I])^2}.
\end{equation}
Equating this to zero and solving for $\sigma$ gives
\begin{equation}
\sigma^* = \beta \left(\E[I] + \frac{1}{\mathrm{SNR}} + \sqrt{\mathrm{Var}(I)}\right).
\end{equation}
Substituting this value of $\sigma$ into the upper bound gives the theorem.

\section{Proof of Theorem \ref{thm:mmse-cheby}}
\label{sec:proof-cheby2}

Write $I = I_{\Nr}$.  Then,
\begin{eqnarray}
1-P_{\rm out}^{\mathrm{mmse}} & \leq & \prob \left(\frac{S}{I} \geq \beta \right) \nonumber \\
& \leq & \prob \left(\left| \frac{S}{I} - \E \left[ \frac{S}{I} \right] \right| \geq \beta - \E \left[ \frac{S}{I} \right] \right) \nonumber \\
& \leq & \frac{\mathrm{Var}\left(\frac{S}{I}\right)}{\left(\beta - \E \left[ \frac{S}{I}\right]\right)^2}.
\end{eqnarray}
The independence of $S,I$ yields:
\begin{equation}
\E \left[ \frac{S}{I}\right] = \E[S] \E \left[\frac{1}{I} \right], ~
\mathrm{Var}\left(\frac{S}{I}\right) = \E[S^2] \E \left[\frac{1}{I^2}\right] - \E[S]^2 \E \left[\frac{1}{I}\right]^2.
\end{equation}
We know $S \sim \frac{1}{2} \chi^2_{2 n_R}$ and thus $\E[S] = n_R$ and $\E[S^2] = n_R(n_R+1)$.  It remains to find a lower bound on $\E \left[\frac{1}{I}\right]$ and an upper bound on $\E \left[\frac{1}{I^2}\right]$.  Jensen's inequality gives a lower bound on $\E \left[\frac{1}{I}\right]$ since the function $1/x$ is convex:
\begin{equation}
\E \left[\frac{1}{I}\right] \geq \frac{1}{\E[I]} = \frac{1}{(\pi d^2 \lambda)^{\frac{\alpha}{2}} \sum_{i=n_R}^{\infty} \E[T_i^{-\frac{\alpha}{2}}]}.
\end{equation}
Consequently, using $\E[S] = n_R$,
\begin{equation}
\E[S]^2 \E \left[\frac{1}{I}\right]^2 \geq \frac{n_R^2}{(\pi d^2 \lambda)^{\alpha} \left(\sum_{i=n_R}^{\infty} \E[T_i^{-\frac{\alpha}{2}}] \right)^2}.
\end{equation}
By Lemma \ref{lem:3}, the upper bound on $\E \left[ \frac{1}{I^2} \right]$ is:
\begin{equation}
\E \left[ \frac{1}{I^2} \right] \leq   \frac{(\pi d^2 \lambda)^{-\alpha} e^{2 \gamma}}{\left( \sum_{i=n_R}^{\infty} e^{-\frac{\alpha}{2} \psi_0(i)} \right)^2 } =  \frac{(\pi d^2 \lambda)^{-\alpha}}{\left( \sum_{i=n_R}^{\infty} e^{-\left(\gamma+\frac{\alpha}{2} \psi_0(i) \right)} \right)^2 }.
\end{equation}
Thus
\begin{equation}
\mathrm{Var}\left(\frac{S}{I}\right) \leq \frac{n_R(n_R+1)}{(\pi d^2 \lambda)^{\alpha} \left( \sum_{i=n_R}^{\infty} e^{-\left(\gamma+\frac{\alpha}{2} \psi_0(i) \right)} \right)^2 } - \frac{n_R^2}{(\pi d^2 \lambda)^{\alpha} \left(\sum_{i=n_R}^{\infty} \frac{\Gamma \left(i-\frac{\alpha}{2}\right)}{\Gamma(i)} \right)^2}.
\end{equation}
which completes the proof.

\begin{lemma}
\label{lem:1}
The function $f_{a,b}(x,y) = \left(\sum_i x_i^a y_i^b \right)^{-2}$ with parameters $a,b$ is a concave function in $x,y$ for all $a,b$ such that $a+b \in \left[-\frac{1}{2},0\right]$.
\end{lemma}

The proof of the lemma is as follows.  View $f_{a,b}(x,y)$ as a composition of functions
\begin{equation}
f_{a,b}(x,y) = h(g_{a,b}(x,y)), ~ h(g) = \frac{1}{g^2}, ~~ g_{a,b}(x,y) = \sum_i x_i^a y_i^b.
\end{equation}
Note $h$ is convex decreasing with derivatives
\begin{equation}
h'(g) = -\frac{2}{g^3}, ~~  h^{''}(g) = \frac{6}{g^4}.
\end{equation}
while $g$ has derivatives:
\begin{equation}
\begin{array}{lll}
\frac{\partial g_{a,b}(x,y)}{\partial x_i} = a x_i^{a-1} y_i^b &
\frac{\partial g_{a,b}(x,y)}{\partial y_i} = b x_i^a y_i^{b-1} \\
\frac{\partial^2 g_{a,b}(x,y)}{\partial x_i^2} = a(a-1) x_i^{a-2} y_i^b &
\frac{\partial^2 g_{a,b}(x,y)}{\partial y_i^2} = b (b-1) x_i^a y_i^{b-2} &
\frac{\partial^2 g_{a,b}(x,y)}{\partial x_i \partial y_i} = a b x_i^{a-1} y_i^{b-1} \\
\frac{\partial^2 g_{a,b}(x,y)}{\partial x_i x_j}=0 & \frac{\partial^2 g_{a,b}(x,y)}{\partial y_i y_j}=0 & \frac{\partial^2 g_{a,b}(x,y)}{\partial x_i y_j} = 0
\end{array}
\end{equation}
Using
\begin{eqnarray}
\frac{\partial f_{a,b}(x,y)}{\partial z_i} &=& h'(g_{a,b}(x,y)) \frac{\partial g_{a,b}(x,y)}{\partial z_i}, ~ z_i \in \{x_i,y_i\}  \nonumber \\
\frac{\partial^2 f_{a,b}(x,y)}{\partial z_i \partial z_j}  &=& h^{''}(g_{a,b}(x,y)) \frac{\partial g_{a,b}(x,y)}{\partial z_i} \frac{\partial g_{a,b}(x,y)}{\partial z_j}  + h'(g_{a,b}(x,y)) \frac{\partial^2 g_{a,b}(x,y)}{\partial z_i z_j}, ~ z_i \in \{x_i,y_i\}, ~ z_j \in \{x_j,y_j\}.
\end{eqnarray}
we compute the quadratic form as $(x,y) \nabla^2 f_{a,b}(x,y) (x,y) = $
\begin{eqnarray}
&=& \sum_i x_i \left( \frac{\partial^2 f_{a,b}(x,y)}{\partial x_i^2} x_i + \frac{\partial^2 f_{a,b}(x,y)}{\partial x_i y_i} y_i + \sum_{j \neq i} \frac{\partial^2 f_{a,b}(x,y)}{\partial x_i x_j} x_j + \sum_{j \neq i} \frac{\partial^2 f_{a,b}(x,y)}{\partial x_i y_j} y_j \right) + \nonumber \\
& & \sum_i y_i \left( \frac{\partial^2 f_{a,b}(x,y)}{\partial y_i^2} y_i + \frac{\partial^2 f_{a,b}(x,y)}{\partial x_i y_i} x_i + \sum_{j \neq i} \frac{\partial^2 f_{a,b}(x,y)}{\partial y_i y_j} y_j + \sum_{j \neq i} \frac{\partial^2 f_{a,b}(x,y)}{\partial y_i x_j} x_j \right) \nonumber \\
&=& \frac{2(a + b) (1 + 2 (a + b))}{g_{a,b}(x,y)^2}
\end{eqnarray}
The sign of the quadratic form depends upon $a,b$ only through the sum $a+b$.  In particular, the quadratic form is negative for all $a,b$ such that $a+b \in \left[-\frac{1}{2},0\right]$, and thus $f_{a,b}(x,y)$ is concave for these values.

\begin{lemma}
\label{lem:2}
For any $a,b$ with $a+b \in \left[-\frac{1}{2},0\right]$ and $b \leq -1$:
\begin{equation}
\E \left[ \frac{1}{I^2} \right] \leq \frac{(\pi d^2 \lambda)^{-\alpha}}{\left( \sum_{i=n_R}^{\infty} \E \left[ T_i^{-\frac{\alpha}{2a}} \right]^a \E \left[ H_i^{\frac{1}{b}} \right]^b \right)^2 }.
\end{equation}
\end{lemma}

The proof of the lemma is as follows.    Define $X_i = T_i^{-\frac{\alpha}{2a}}$ and $Y_i = H_i^{\frac{1}{b}}$.  Then:
\begin{equation}
\E \left[ \frac{1}{I^2} \right]
= \E \left[ \frac{(\pi d^2 \lambda)^{-\alpha}}{ \left( \sum_{i=n_R}^{\infty} T_i^{-\frac{\alpha}{2}} H_i\right)^2 } \right]
= \E \left[ \frac{(\pi d^2 \lambda)^{-\alpha}}{ \left( \sum_{i=n_R}^{\infty} \left(T_i^{-\frac{\alpha}{2a}}\right)^a \left(H_i^{\frac{1}{b}}\right)^b \right)^2} \right]
= (\pi d^2 \lambda)^{-\alpha} \E[f_{a,b}(X,Y)].
\end{equation}
The proof follows by Jensen's inequality and Lemma \ref{lem:1}.  The constraint $b \leq -1$ follows from:
\begin{equation}
\E \left[ H^c \right] = \left\{ \begin{array}{ll}
\Gamma(1+c), \; & c \geq -1 \\
\infty, \; & \mbox{else} \end{array} \right.
\end{equation}
Thus to ensure $\E \left[ H_i^{\frac{1}{b}} \right] < \infty$ we must have $1/b \geq -1$, or $b \leq -1$.

\begin{lemma}
\label{lem:3}
An upper bound on $\E[1/I^2]$ is
\begin{equation}
\E \left[ \frac{1}{I^2} \right] \leq   \frac{(\pi d^2 \lambda)^{-\alpha} e^{2 \gamma}}{\left( \sum_{i=n_R}^{\infty} e^{-\frac{\alpha}{2} \psi_0(i)} \right)^2 } =  \frac{(\pi d^2 \lambda)^{-\alpha}}{\left( \sum_{i=n_R}^{\infty} e^{-\left(\gamma+\frac{\alpha}{2} \psi_0(i) \right)} \right)^2 }.
\end{equation}
where $\gamma$ is the Euler-Mascheroni constant and $\psi_0(i)$ is the poly-Gamma function.
\end{lemma}

The proof of the lemma is as follows.   Using
\begin{equation}
\E \left[ T_i^{-\frac{\alpha}{2a}} \right]^a = \left( \frac{\Gamma \left(i-\frac{\alpha}{2a}\right)}{\Gamma(i)}\right)^a, ~~ \E \left[ H_i^{\frac{1}{b}} \right]^b = \Gamma\left(1+\frac{1}{b}\right)^b,
\end{equation}
and Lemma \ref{lem:2} we pose the following optimization problem (to make the bound as tight as possible):
\begin{equation}
\max_{a,b} \left\{ \left.  \Gamma\left(1+\frac{1}{b}\right)^b \sum_{i=n_R}^{\infty} \left( \frac{\Gamma \left(i-\frac{\alpha}{2a}\right)}{\Gamma(i)}\right)^a  \right| b \leq -1, ~ a+ b \geq -\frac{1}{2}, ~ a+b \leq 0 \right\}.
\end{equation}
For any $b$ the optimal choice of $a$ is such that $a+b=-1/2$ and for this (or any) choice of $a$ the objective is unbounded as $b \to -\infty$.  Consider the limit as $b \to -\infty$ and $a=-b-\frac{1}{2}$.  The proof follows from the facts that
\begin{equation}
\lim_{b \to -\infty} \Gamma\left(1+\frac{1}{b}\right)^b = e^{-\gamma},
\end{equation}
and
\begin{equation}
\lim_{b \to -\infty} \left( \frac{\Gamma \left(i-\frac{\alpha}{2(-b-\frac{1}{2})}\right)}{\Gamma(i)}\right)^{-b-\frac{1}{2}} = e^{-\frac{\alpha}{2} \psi_0(i)},
\end{equation}
where $\gamma$ is the Euler-Mascheroni constant and $\psi_0(i)$ is the zero-order poly-gamma function.



\end{document}